\documentclass[11pt,a4paper]{article}
\usepackage{amsmath,amssymb,amsthm,newlfont,graphicx,amscd,bbm,enumerate,galois,mathrsfs,xypic,geometry}
\usepackage{simplewick}
\usepackage{tcolorbox}
\usepackage{empheq}
\usepackage{indentfirst,hyperref}

\usepackage{color}

\newcommand{\clr}{\textcolor[rgb]{1.00,0.00,0.00}}

\newcommand{\lmd}{\lambda}

\newcommand{\Lmd}{\Lambda}

\newcommand{\omg}{\omega}

\newcommand{\p}{\partial}

\newcommand{\afa}{\alpha}

\newcommand{\veps}{\varepsilon}
\newcommand{\fai}{\varphi}
\newcommand{\rme}{{\mathrm e}}
\newcommand{\rmT}{{\mathrm T}}
\DeclareMathOperator{\Res}{Res}
\newcommand{\mcalA}{\mathcal{A}}

\newcommand{\mcalF}{\mathcal{F}}

\newcommand{\mcalJ}{\mathcal{J}}

\newcommand{\mcalL}{\mathcal{L}}
\newcommand{\mcalM}{\mathcal{M}}
\newcommand{\mcalP}{\mathcal{P}}
\newcommand{\mcalQ}{\mathcal{Q}}

\newcommand{\bbZ}{\mathbb{Z}}

\newcommand{\bfu}{\mathbf{u}}

\newcommand*{\pp}[1]
{\frac{\partial   }
	{\partial #1}
}

\newcommand*{\pfrac}[2]
{\frac{\partial #1}
	{\partial #2}
}

\newcommand*{\pair}[2]                   
{
	\left\langle
	#1,#2
	\right\rangle
}

\newcommand*{\Bigset}[2]
{
	\left\{ #1 \,\middle|\, #2 \right\}             
}

\newcommand{\beq}{\begin{equation}}
	\newcommand{\eeq}{\end{equation}}

\DeclareMathOperator*{\res}{Res}

\DeclareMathOperator{\Der}{Der}
\DeclareMathOperator{\grad}{grad}

\DeclareMathOperator{\td}{d\!}

\newtheorem{thm}{Theorem}[section]

\newtheorem{rmk}[thm]{Remark}

\newtheorem{lem}[thm]{Lemma}

\newtheorem{prop}[thm]{Proposition}
\newtheorem{defn}{Definition}[section]
\newtheorem{ex}{Example}[section]

\numberwithin{equation}{section}

\allowdisplaybreaks[4]

\begin{document}
	\title {{ Asymmetric rational reductions of 2D-Toda hierarchy \\and a generalized Frobenius manifold}
	\footnotetext{{\footnotesize$^*$Correspondence should be addressed to Qiulan Zhao; } \\
		{\footnotesize Electronic mail:} {\footnotesize
			qlzhao@sdust.edu.cn}}}
	\author{{
		Haonan Qu$^{1}$,~ Qiulan Zhao$^{*2}$} \\
	[1em]
		\footnotesize $^{1}$ School of Mathematical Sciences, Peking University,\\
		\footnotesize Beijing, 100871, P. R. China.\\
		[1em]
		\footnotesize $^{2}$ College of Mathematics and Systems Science, Shandong University of Science and Technology,\\
		\footnotesize Qingdao, 266590, Shandong, P. R. China. \\
	}
	\maketitle
	\begin{abstract}
We study the local bihamiltonian structures of the asymmetric
rational reductions of the 2D-Toda hierarchy (RR2T) of types $(2,1)$ and $(1,2)$ at the full-dispersive level,
and construct a three-dimensional generalized Frobenius manifold with non-flat unity associated with the $(2,1)$-type.
Furthermore, we explicitly relate the $(2,1)$-type RR2T to the bi-graded Toda and constrained KP hierarchies
via linear reciprocal and Miura-type transformations.
	\end{abstract}
	\tableofcontents
	\section{Introduction}
	\par The two-dimensional Toda equation
	\begin{equation}
		\left(\partial_{x}^2-\partial_{t}^2\right)x_n=\rme^{x_{n+1}}-2\rme^{x_n}+\rme^{x_{n-1}},
	\end{equation}
	possesses interesting structures and properties \cite{1,3}, and its commuting flows generate the famous \textit{2D-Toda hierarchy} \cite{Toda hierarchy} with a tri-Hamiltonian structure \cite{trh -toda}. This hierarchy and its various reductions play important roles in the field of mathematical physics, such as orthogonal polynomials and random matrix theory \cite{random permutations,Matrix integrals}. Takasaki studied a reduction of the 2D Toda hierarchy at the full dispersive level, in which the Lax operator has the form \cite{Takasaki}
	\begin{equation}\label{TLO}
		\mathbb L = \left(1 + \beta_1 \Lambda^{-1} + \cdots + \beta_N \Lambda^{-N}\right)^{-1}\Lmd^{1-N}
\left(\Lambda^N + \alpha_{1} \Lambda^{N-1} + \cdots + \alpha_N\right)
	\end{equation}
for some $N\in\bbZ_+$,
	where the shift operator $\Lmd=\rme^{\veps\p_x}$ satisfies $\Lmd y (x)=y (x+\veps)$.
	This decomposition characterizes the generalized Ablowitz–Ladik hierarchy viewed as a subsystem of the Toda hierarchy, and it has been shown to be an important integrable structure of topological string theory on $N$-th generalized conifold \cite{Takasaki}.
	Subsequently, Brini and his collaborators constructed a two-parameter family of the so-called \textit{rational reductions of 2D-Toda hierarchy}
(abbreviated as RR2T), and established their characterization in terms of block Toeplitz matrices for the
	associated factorization problem \cite{rational-reduction}. The distinctive feature of the RR2T of type
	$(a, b)\in\bbZ_+^2$ is that the corresponding Lax operator has the form
	\begin{equation}\label{ab type RR2T}
		L_{(a,b)} = \left(  1 + \beta_1 \Lambda^{-1} + \cdots + \beta_b \Lambda^{-b}\right)^{-1}
\left(
  \Lambda^a + \alpha_{1} \Lambda^{a-1} + \cdots + \alpha_a
\right).
	\end{equation}
	They generalized the relation between the Ablowitz--Ladik hierarchy and Gromov–Witten theory by proving an analogous mirror theorem for the general rational reduction \cite{rational-reduction}, however their discussion was restricted to the dispersionless level.
Many well-known integrable hierarchies arise from further reductions of the RR2T, including the bi-graded Toda hierarchy \cite{extend bi Toda} and the $q$-deformed KdV hierarchy \cite{q-deformed GD}. 

 If the parameters $a$ and $b$ in \eqref{ab type RR2T} are equal,
 the corresponding RR2T is called \textit{symmetric}; otherwise, it is called \textit{asymmetric}.
 A typical example of the symmetric RR2T is the case $a=b=1$,
 in which the corresponding integrable hierarchy is the Ablowitz--Ladik hierarchy \cite{AL, Brini 2012, AL-triham},
 also known as the relativistic Toda hierarchy \cite{Kupershmidt, rtoda}.
 As one of the most important differential-difference hierarchies,
it has been extensively studied \cite{Geng,suris,AL symmetries1,AL symmetries2},
and was shown in \cite{AL-triham} to admit a local tri-Hamiltonian structure.
The relation between a particular tau function of
this hierarchy and the generating function of Gromov--Witten invariants of local $\mathbb{CP}^1$ was studied in \cite{GW, Brini 2012,GW2}.
Moreover, Brini and his collaborators introduced a generalized Frobenius manifold $M_{\text{AL}}$ with non-flat unity
corresponding to the Ablowitz--Ladik hierarchy in \cite{Brini 2012}.
Subsequently, motivated by this $M_{\text{AL}}$,
the Dubrovin–Zhang theory for the topological deformation of the Principal Hierarchy of the
class of generalized Frobenius manifolds with non-flat unity was established in \cite{Liu2024, GFM}.
Afterwards, an extension of the Ablowitz--Ladik hierarchy was given and shown to be equivalent to the topological deformation of the Principal Hierarchy of $M_{\text{AL}}$ \cite{extend AL}.
Therefore, the Ablowitz–Ladik hierarchy has been formulated within the framework of Dubrovin–Zhang theory under the four-axiom setting
\cite{Normal forms, Liu lecture notes} with special conditions.

In contrast to the symmetric case, in this paper we investigate two typical examples of asymmetric RR2T:
the RR2T of type $(2,1)$ with Lax operator
	\begin{align}\label{250905-1432}
		L&=\left(1-w\Lmd^{-1}\right)^{-1}\left(\Lmd^2+u\Lmd+v\right),
	\end{align}
and that of type $(1,2)$ with Lax operator
	\begin{align} \label{Lax type 1,2}
		\mcalL&=\
		\left(1+U\Lambda^{-1}+V\Lambda^{-2}\right)^{-1}\left(\Lambda+W\right).
	\end{align}
We show that each of them admits a local bihamiltonian structure,
and that the two hierarchies, together with their corresponding bihamiltonian structures, are related by a certain Miura transformation and spatial reflection.
We also construct a generalized Frobenius manifold $M$ with non-flat unity associated with the $(2,1)$-type RR2T,
and show that its dispersionless flows belong to the Principal Hierarchy of $M$.
Moreover, there are two generalized Legendre transformations \cite{LiuQuZhang,GLt}
that relate $M$ to the Frobenius manifolds corresponding to the bi-graded Toda and the constrained KP hierarchies.
These transformations motivate us to relate the RR2T of type $(2,1)$ to the bi-graded Toda hierarchy and the constrained KP hierarchy
via certain linear reciprocal and Miura-type transformations.

The paper is organized as follows. In Sect.\,2, we derive the bihamiltonian formalisms of the two hierarchies mentioned above,
and relate these two hierarchies via a certain Miura-type transformation up to spatial reflection.
Moreover, we also show that all the central invariants of the bihamiltonian structures of the $(2,1)$-type RR2T are equal to $\frac{1}{24}$.
In Sect.\,3, we study the generalized Frobenius manifold $M$ associated with the $(2,1)$-type RR2T.
In Sect.\,4, we explicitly construct the linear reciprocal and Miura-type transformations that relate the RR2T of type $(2,1)$ to the bi-graded Toda and the constrained KP hierarchies. In Sect.\,5, we summarize the main conclusions and discuss possible directions for future research.

	\section{Asymmetric RR2T and their bihamiltonian structures}

Let us begin with the definition and basic properties of the RR2T of type $(2,1)$.

\subsection{The definition of the RR2T of type $(2,1)$}
\label{subsection: defo of 2,1 type}

Note that the Lax operator of the $(2,1)$-type RR2T \eqref{250905-1432} has the form
\begin{equation}\label{Lax-L}
\begin{split}
  L&=\,(1-w\Lmd^{-1})^{-1}(\Lmd^2+u\Lmd+v),\\
  &=\,
    \Lmd^2 + (u+w)\Lmd + (v+u^-w+w^-w)+\cdots,
\end{split}
\end{equation}
here the shift operator $\Lmd=\rme^{\veps\p_x}$,
unknown functions $(u^1,u^2,u^3):=(u,v,w)$, and
for any function $f$, we introduce the short notations
\begin{equation}
  f^+=\Lmd f,\qquad f^-=\Lmd^{-1}f.
\end{equation}

If we denote the operators
\begin{equation}\label{Lax-AB}
  A=\Lmd^2+u\Lmd+v,\quad
  B=1-w\Lmd^{-1},
\end{equation}
then $L=B^{-1}A$. We also introduce
\begin{equation}\label{Lax-Ltil}
\begin{split}
\tilde L := AB^{-1} = \Lmd^2 + (u+w^{++})\Lmd+(v+uw^++w^+w^{++}) +\cdots .
\end{split}
\end{equation}
For any operator of the form $X=\sum_{k\in\bbZ} a_k\Lmd^k$,
we denote
\begin{equation}
  X_+=\sum_{k\geq 0}a_k\Lmd^k,\qquad
  X_-=\sum_{k\leq -1}a_k\Lmd^k.
\end{equation}

\begin{defn}
  The RR2T of type $(2,1)$ consists of the following flows:
\begin{align}
  \veps\pfrac{A}{t^{2,k}}
&=\, \frac{2^{k+\frac12}}{(2k+1)!!}
   \left(
     \left(
       \tilde L^{k+\frac12}
     \right)_+A
    -A\left(
        L^{k+\frac12}
      \right)_+
   \right) , \label{GAL-1}
\\
  \veps\pfrac{B}{t^{2,k}}
&=\, \frac{2^{k+\frac12}}{(2k+1)!!}
   \left(
     \left(
       \tilde L^{k+\frac12}
     \right)_+B
    -B\left(
        L^{k+\frac12}
      \right)_+
   \right) , \label{GAL-2}
\\
  \veps\pfrac{A}{t^{3,k}}
&=\, \frac{1}{(k+1)!}
   \left(
     \left(
       \tilde L^{k+1}
     \right)_+A
    -A\left(
        L^{k+1}
      \right)_+
   \right) , \label{GAL-3}
\\
  \veps\pfrac{B}{t^{3,k}}
&=\, \frac{1}{(k+1)!}
   \left(
     \left(
       \tilde L^{k+1}
     \right)_+B
    -B\left(
        L^{k+1}
      \right)_+
   \right) , \label{GAL-4}
\\
  \veps\pfrac{A}{t^{0,-k-1}}
&=\, (-1)^kk!
   \left(
     \left(
       \tilde M^{k+1}
     \right)_-A
    -A\left(
        M^{k+1}
      \right)_-
   \right) , \label{GAL-5}
\\
  \veps\pfrac{B}{t^{0,-k-1}}
&=\, (-1)^kk!
   \left(
     \left(
       \tilde M^{k+1}
     \right)_- B
    -B\left(
        M^{k+1}
      \right)_-
   \right) \label{GAL-6}
\end{align}
for all $k\geq 0$.  Here $\{t^{2,k}, t^{3,k}, t^{0,-k-1}\}_{k\geq 0}$ are the set of time variables,
the operators
\[
  L^{\frac12}=\Lmd+\sum_{r\geq 0}c_r\Lmd^{-r},\qquad
  \tilde L^{\frac12} = \Lmd + \sum_{r\geq 0} \tilde c_r\Lmd^{-r}
\]
are uniquely determined by $(L^{\frac12})^2=L$,
$(\tilde L^{\frac12})^2=\tilde L$ respectively, and
\begin{align}
  M&=\, A^{-1}B = -\frac wv\Lmd^{-1}+\sum_{r\geq 0}m_r\Lmd^r, \label{Lax-M}\\
  \tilde M &=\, BA^{-1} =  -\frac{w}{v^-}\Lmd^{-1}
    +\sum_{r\geq 0}\tilde m_r\Lmd^r
\end{align}
for certain functions $m_r$ and $\tilde m_r$.
\end{defn}

It is clear that \eqref{GAL-1}--\eqref{GAL-6} are equivalent to the following Lax equations for $k\geq 0$:
\begin{align}
  \veps\pfrac{L}{t^{2,k}}
&=\,
  \frac{2^{k+\frac12}}{(2k+1)!!}
    \left[
      \left(L^{k+\frac12}\right)_+, L
    \right],  \label{Lax RR2T-1}
\\
  \veps\pfrac{L}{t^{3,k}}
&=\,
  \frac{1}{(k+1)!}
    \left[
      \left(L^{k+1}\right)_+, L
    \right], \label{Lax RR2T-2}
\\  \veps\pfrac{L}{t^{0,-k-1}}
&=\,
  (-1)^kk!
    \left[
      \left(M^{k+1}\right)_-, L
    \right].
\end{align}

\begin{rmk}
  The RR2T of type $(2,1)$ admits a family of additional symmetries, which are denoted by
  $\{\pp{t^{0,k}}, \pp{t^{1,k}}\}_{k\geq 0}$, where $\pp{t^{0,0}}=\p_x$.
  See Remark \ref{rmk:extended flows} below.
  Moreover, we will identify the spatial variable $x$ with the new time variable $t^{0,0}$, i.e.
\begin{equation}\label{flow-ppx}
  x=t^{0,0},\qquad \p_{x}=\pp{t^{0,0}}.
\end{equation}
\end{rmk}

Now we are to express the first few flows of this hierarchy explicitly.
Note that
\begin{align*}
  L^{\frac12}&=\,
\Lmd + \frac{1}{\Lmd+1}(u+w)
+\frac{1}{\Lmd+1}
\left(
  v+u^-w+w^-w-
  \left(
    \frac{1}{\Lmd+1}(u+w)
  \right)^2
\right)\Lmd^{-1}+\cdots ,
\\
  \tilde L^{\frac12}&=\,
\Lmd + \frac{1}{\Lmd+1}(u+w^{++})
+\frac{1}{\Lmd+1}
\left(
  v+uw^++w^+w^{++} -
  \left(
    \frac{1}{\Lmd+1}(u+w^{++})
  \right)^2
\right)\Lmd^{-1}+\cdots ,
\end{align*}
here the differential operator
$
  \frac{1}{\Lmd+1} = \frac{1}{\rme^{\veps\p_x}+1}
=
  \frac12 - \frac{\veps}{4}\p_x
 +\frac{\veps^3}{48}\p_x^3-\frac{\veps^5}{480}\p_x^5+O(\veps^7).
$
Then using \eqref{GAL-1}--\eqref{GAL-6} and by straightforward calculation, we obtain
\begin{align}
  \pfrac{u}{t^{2,0}}
&=\,
  \frac{\sqrt{2}}{\veps}
  \left(
  (v^+-v)+u\frac{\Lmd-1}{\Lmd+1}(w^+-u)
  \right) \notag\\
&=\,
  \frac{1}{\sqrt{2}}\left(-u u_x +2 v_x + u w_x\right)+
  \frac{\veps}{\sqrt{2}}\left(v_{xx} + u w_{xx}\right) + O(\veps^2),
\label{rad-flow-u}\\
  \pfrac{v}{t^{2,0}}
&=\,
  \frac{\sqrt{2}}{\veps}
  v(w^+-w)
=
 \sqrt{2}\, v w_x
 + \frac{\veps}{\sqrt{2}}v w_{xx} + O(\veps^2)
,
\label{rad-flow-v}\\
  \pfrac{w}{t^{2,0}}
&=\,
  \frac{\sqrt{2}}{\veps}
  w\frac{\Lmd-1}{\Lmd+1}
  \left(w^++w+w^-+u^-\right) \notag\\
&=\,
  \frac{1}{\sqrt{2}}\left(w u_x + 3 w w_x\right)
  -\frac{\veps}{\sqrt{2}}w u_{xx} + O(\veps^2),
\label{rad-flow-w}
\end{align}
and
\begin{align}
  \veps\pfrac{u}{t^{3,0}}
&=\,
  uw^+w^{++} + v^+w^{++} - uww^+-vw,
\label{positive-flow-u}\\
  \veps\pfrac{v}{t^{3,0}}
&=\,
  v\left(
    uw^+ + w^+w^{++} - u^-w - w^-w
  \right),
\label{positive-flow-v}\\
  \veps\pfrac{w}{t^{3,0}}
&=\,
  w\left(
    uw^+ + w^+w^{++} + v
   -u^{--}w^- -w^{--}w^- - v^-
  \right),\label{positive-flow-w}
\end{align}
\begin{equation}\label{negative-t-0-(-1)}
\veps\pfrac{u}{t^{0,-1}}
= \frac{w^{++}}{v^{++}}-\frac{w}{v^-},
\quad
\veps\pfrac{v}{t^{0,-1}}
=
  \frac{uw^+}{v^+}-\frac{u^-w}{v^-},
\quad
\veps\pfrac{w}{t^{0,-1}}
=\frac{w}{v^-}-\frac{w}{v}.
\end{equation}

\subsection{Bihamiltonian formalism of the RR2T of type $(2,1)$}
\label{subsection: biham formalism of 2,1}

Now we are to derive the bihamiltonian formalism of the $(2,1)$-type RR2T.
Consider the following skew-symmetric
operator-valued matrix
\begin{equation}\label{HamP1}
  \mcalP_1
=
  \begin{pmatrix}
    \Lmd-\Lmd^{-1} & w\Lmd^{-1}-\Lmd^2w & 0 \\
    w\Lmd^{-2}-\Lmd w & w\Lmd^{-1}u-u\Lmd w & (\Lmd-1)w \\
    0 & w(1-\Lmd^{-1}) & 0
  \end{pmatrix},
\end{equation}
we declare that this is a Hamiltonian operator.
In fact, consider the following new coordinates
\begin{equation}\label{full-genera-z}
		\left\{
		\begin{aligned}
			z^1&= v+u^{-}w+w^{-}w,\\
			z^2&= \frac{\sqrt{2}}{\Lmd+1}(u+w),\\
			z^3&= \log w,
		\end{aligned}
		\right.
	\end{equation}
then the Jacobian $\mcalJ=(J^\afa_\beta)=
	\left(
	\sum\limits_{k\geq 0}
	\pfrac{z^\afa}{u^{\beta,k}}\p_x^k
	\right)$
and its adjoint $\mcalJ^\dag$ have the form
	\begin{equation*}
		\mcalJ
		=
		\begin{pmatrix}
			w\Lmd^{-1} & 1 & u^-+w^-+w\Lmd^{-1} \\[4pt]
			\frac{\sqrt{2}}{\Lmd+1} & 0 & \frac{\sqrt{2}}{\Lmd+1} \\[4pt]
			0 & 0 & \frac 1w
		\end{pmatrix},\quad
		\mcalJ^\dag =
		\begin{pmatrix}
			\Lmd w & \frac{\sqrt{2}\Lmd}{\Lmd+1} & 0 \\[4pt]
			1 & 0 & 0 \\[4pt]
			u^-+w^-+\Lmd w & \frac{\sqrt{2}\Lmd}{\Lmd+1} & \frac 1w
		\end{pmatrix},
	\end{equation*}
where $(u^1, u^2, u^3)=(u,v,w)$ are the old coordinates, and jet variables $u^{\beta,k}:=\p_x^k u^{\beta}$.
Then it can be verified that
	\begin{align}\label{const-HamP1}
		\mcalJ\mcalP_1\mcalJ^\dag &=\,
		\begin{pmatrix}
			0 & 0 & \Lmd-1 \\[4pt]
			0 & 2\frac{\Lmd-1}{\Lmd+1} & 0 \\[4pt]
			1-\Lmd^{-1} & 0 & 0
		\end{pmatrix},
	\end{align}
which is a constant Hamiltonian operator with respect to the new coordinates $(z^1, z^2, z^3)$.

\begin{thm}\label{thm:HamFormalism-1}
  The RR2T of $(2,1)$-type \eqref{GAL-1}--\eqref{GAL-6}
  can be represented by the following Hamiltonian systems:
  \begin{equation}\label{HamFormalism-1}
    \veps\pp{t^{i,k}}
    \begin{pmatrix}
      u \\[2pt]
      v \\[2pt]
      w
    \end{pmatrix}
  =
    \mcalP_1
    \begin{pmatrix}
      \delta H_{i,k}/\delta u \\[2pt]
      \delta H_{i,k}/\delta v \\[2pt]
      \delta H_{i,k}/\delta w
    \end{pmatrix}
  \end{equation}
for $i\in\{2,3\}$, $k\geq 0$ or $i=0$, $k\leq -1$.
Here the Hamiltonian operator $\mcalP_1$ is given in \eqref{HamP1},
and the Hamiltonians have the expressions
\begin{align}\label{AL-H2k}
  H_{2,k} &=\,
  \frac{2^{k+\frac 32}}{(2k+3)!!}
  \int\Res\left(L^{k+\frac 32}\right)\td x,\quad k\geq -1,
\\
  H_{3,k} &=\,
  \frac{1}{(k+2)!}
  \int\Res\left(L^{k+2}\right)\td x,\quad k\geq -1,
\\
  H_{0,-k-1} &=\,
  (-1)^k(k-1)!
  \int\Res\left(M^k\right)\td x,\quad k\geq 1,  \label{AL-H0-k-1}
\\
  H_{0,-1} &=\, \int \log\frac wv \td x,  \label{H-0,-1}
\end{align}
where $M=A^{-1}B$ is defined as in \eqref{Lax-M}.
\end{thm}

We note that for an operator of the form $X=\sum_{s\in\bbZ}a_s\Lmd^s$,
its residue $\Res X = a_0$,
and for a local functional $H=\int h(u, u_x, u_{xx}, \dots)\td x$, its variational derivatives
\[
  \frac{\delta H}{\delta u^\afa}
=
  \sum_{s\geq 0}
    (-\p_x)^s\pfrac{H}{u^{\afa,s}}.
\]

To prove the above theorem, we introduce some notations and lemmas.
For the operators $L, \tilde L, B$ defined in \eqref{Lax-L}--\eqref{Lax-Ltil}, denote
\begin{equation}\label{def-of-coef-ab}
  L^k = \sum_{s\leq 2k}a^k_s \Lmd^s, \qquad
  \tilde L^k = \sum_{s\leq 2k}\tilde a^k_s \Lmd^s, \qquad
  L^kB^{-1}  = \sum_{s\leq 2k} b^k_s \Lmd^s
\end{equation}
for $k\geq 0$, and $a^k_s=\tilde a^k_s = b^k_s = 0$ if $s>2k$,
then \eqref{GAL-3}--\eqref{GAL-4} can be rewritten as
\begin{align}
  \veps\pfrac{u}{t^{3,k}}
&=\,
  \frac{1}{(k+1)!}
  \left(
    u\tilde a^{k+1}_0 - u\left(a^{k+1}_0\right)^+
   +v^+\tilde a_1^{k+1} - v a^{k+1}_1
  \right),
\label{Dut3k}\\
  \veps\pfrac{v}{t^{3,k}}
&=\,
  \frac{1}{(k+1)!}
  \left(
    v\tilde a^{k+1}_0 - v a^{k+1}_0
  \right),
\label{Dvt3k}\\
  \veps\pfrac{w}{t^{3,k}}
&=\,
  \frac{1}{(k+1)!}
  \left(
    w\tilde a^{k+1}_0 - w\left(a^{k+1}_0\right)^-
  \right).
\label{Dwt3k}
\end{align}

\begin{lem}
  The following identities hold true for each $k\geq 0$:
  \begin{align}
     \tilde a^{k}_{-1}
   + u\tilde a^{k}_0
   + v^+\tilde a^{k}_1
  &=\,
      (a^{k}_{-1} )^{++}
    +u (a^{k}_0 )^+
    +v a^k_1, \label{coef-relation-1}
  \\
    \tilde a^k_{-1} - w\tilde a^k_0
  &=\,
    a^k_{-1} - w (a^k_0)^-,  \label{coef-relation-2}
  \\
    \tilde a^k_0 - w^+\tilde a^k_1
  &=\,
    a^k_0 - w (a^k_1)^-, \label{coef-relation-3}
  \end{align}
\begin{equation}
  \left\{
  \begin{split}
    a^k_1 &=\, b^k_1 - w^{++}b^k_2, \\
    a^k_0 &=\, b^k_0 - w^+ b^k_1, \\
    a^k_{-1} &=\, b^k_{-1} - w b^k_0,
  \end{split}
  \right.
\qquad\quad
  \left\{
  \begin{split}
    \tilde a^k_1 &=\, b^k_1 - w (b^k_2)^-, \\
    \tilde a^k_0 &=\, b^k_0 - w (b^k_1)^-, \\
    \tilde a^k_{-1} &=\, b^k_{-1} - w (b^k_0)^-.
  \end{split}
  \right. \label{coef-relation-ab}
\end{equation}
\end{lem}
\begin{proof}
  The equations \eqref{coef-relation-1}--\eqref{coef-relation-3}
  can be derived from $\tilde L^k A= A L^k$ and $\tilde L^k B = B L^k$,
  while the equations \eqref{coef-relation-ab} follow from
  $(L^k B^{-1})B=L^k$ and $B(L^kB^{-1})=\tilde L^k$.
\end{proof}

Moreover, by using $\tilde L^{k+1}=A(L^k B^{-1})$ and $L^{k+1}=(L^k B^{-1})A$,
we obtain the following recursive relations for each $k\geq 0$:
\begin{align}
  \tilde a^{k+1}_1 &=\,
    (b^k_{-1})^{++} + u(b^k_0)^+ + vb^k_1, \label{coef-recur-1}\\
  a^{k+1}_1 &=\,
    b^k_{-1} + ub^k_0 + v^+b^k_1.  \label{coef-recur-2}
\end{align}

Let us compute the variational derivatives that appear on the right-hand side of \eqref{HamFormalism-1}.
\begin{lem}
  The following identities hold true for each $k\geq 0$:
  \begin{align}\label{var-derivative}
    \frac{\delta H_{3,k}}{\delta u}
  =
    \frac{1}{(k+1)!}(b^{k+1}_{-1})^+,
  \quad
    \frac{\delta H_{3,k}}{\delta v}
  =
    \frac{1}{(k+1)!}b^{k+1}_0,
  \quad
    \frac{\delta H_{3,k}}{\delta w}
  =
    \frac{1}{(k+1)!}(b^{k+2}_1)^-.
  \end{align}
\end{lem}

\begin{proof}
  By using the well-known identity
\[
  \int \Res(XY) \td x = \int \Res(YX) \td x
\]
for any difference operators $X,Y$, we obtain
\begin{align*}
  &\, \delta H_{3,k}
=
  \frac{1}{(k+1)!}
  \int \Res\left(\delta L\cdot L^{k+1}\right) \td x \\
=&\,
  \frac{1}{(k+1)!}
  \int \Res\left(\delta A\cdot L^{k+1}B^{-1} - \delta B\cdot L^{k+2}B^{-1}\right)\td x \\
=&\,
  \frac{1}{(k+1)!}
  \int \left[
    \delta u \cdot \Res\left(\Lmd L^{k+1}B^{-1}\right)
   +\delta v \cdot \Res\left(L^{k+1}B^{-1}\right)
   \right.\\
&\qquad \left.
   +\delta w \cdot \Res\left(\Lmd^{-1}L^{k+2}B^{-1}\right)
   \right] \td x,
\end{align*}
therefore
\begin{align*}
  \frac{\delta H_{3,k}}{\delta u}
=
  \frac{1}{(k+1)!}\Res\left(\Lmd L^{k+1}B^{-1}\right)
=
  \frac{1}{(k+1)!}(b^{k+1}_{-1})^+,
\end{align*}
and the expressions of $\frac{\delta H_{3,k}}{\delta v}$ and $\frac{\delta H_{3,k}}{\delta w}$
can be derived similarly.
The lemma is proved.
\end{proof}

\begin{proof}[Proof of Theorem \ref{thm:HamFormalism-1}]
From \eqref{Dut3k}, \eqref{coef-relation-1}, \eqref{coef-relation-ab} and \eqref{var-derivative}, we obtain
\begin{align*}
  &\,\veps\pfrac{u}{t^{3,k}}
=
  \frac{1}{(k+1)!}
  \left(
    (a^{k+1}_{-1})^{++}-\tilde a^{k+1}_{-1}
  \right) \\
=&\,
  \frac{1}{(k+1)!}
  \left(
    (b^{k+1}_{-1})^{++}
   -w^{++}(b^{k+1}_0)^{++}
   -b^{k+1}_{-1}
   +w(b^{k+1}_0)^-
  \right) \\
=&\,
  (\Lmd-\Lmd^{-1})\frac{\delta H_{3,k}}{\delta u}
 +(w\Lmd^{-1}-\Lmd^2 w)\frac{\delta H_{3,k}}{\delta v}.
\end{align*}
Moreover, by using \eqref{coef-recur-1}--\eqref{coef-recur-2}, we arrive at
\begin{align*}
&\,
  \veps\pfrac{v}{t^{3,k}}
=
  \frac{1}{(k+1)!}
  \left(
    vb^{k+1}_0 - vw(b^{k+1}_1)^-
   -vb^{k+1}_0 + vw^+b^{k+1}_1
  \right)
\\
=&\,
  \frac{1}{(k+1)!}
  \left[
    w^+\left(
      \tilde a^{k+2}_1 - (b^{k+1}_{-1})^{++} - u(b^{k+1}_0)^+
    \right)
    -w\left(
      (a^{k+2}_1)^- - (b^{k+1}_{-1})^- -u^-(b^{k+1}_0)^-
    \right)
  \right]
\\
=&\,
  \frac{1}{(k+1)!}
  \left[
  w^+\left(
      b^{k+2}_1 - (b^{k+1}_{-1})^{++} - u(b^{k+1}_0)^+
    \right)
  -w\left(
      (b^{k+2}_1)^- - (b^{k+1}_{-1})^- -u^-(b^{k+1}_0)^-
    \right)
  \right]
\\=&\,
  (w\Lmd^{-2}-\Lmd w)\frac{\delta H_{3,k}}{\delta u}
 +(w\Lmd^{-1}u-u\Lmd w)\frac{\delta H_{3,k}}{\delta v}
 +(\Lmd-1)w\frac{\delta H_{3,k}}{\delta w},
\end{align*}
and
\begin{align*}
  &\,\veps\pfrac{w}{t^{3,k}}
= \frac{1}{(k+1)!}
  \left(
    w\tilde a^{k+1}_0 - w(a^{k+1}_0)^-
  \right)
\\
=&\,
  \frac{w}{(k+1)!}
  \left[
    \left(
      b^{k+1}_0 - w(b^{k+1}_1)^-
    \right)
   -\left(
     (b^{k+1}_0)^- -w(b^{k+1}_1)^-
   \right)
  \right]
=  w(1-\Lmd^{-1})\frac{\delta H_{3,k}}{\delta v},
\end{align*}
therefore the identities \eqref{HamFormalism-1} hold true for $i=3$ and $k\geq 0$.
The Hamiltonian formalism for the $\pp{t^{2,k}}$-flows and $\pp{t^{0,-k-1}}$-flows of this hierarchy
can be proved in the same way.
We omit the details here. Thus, Theorem \ref{thm:HamFormalism-1} is proved.
\end{proof}

Now let us introduce the operator-valued matrix
\begin{equation}\label{HamP2}
	\mcalP_2=
	\begin{pmatrix}
		\Lmd v-v\Lmd^{-1}+u\frac{1-\Lmd}{1+\Lmd}u& 0 & u\frac{\Lmd-1}{1+\Lmd^{-1}}w\\[4pt]
		0& 0&v(\Lmd-1)w \\[4pt]
		w\frac{1-\Lmd^{-1}}{\Lmd+1}u&w(1-\Lmd^{-1})v &w\frac{\Lmd^{2}-\Lmd^{-1}}{\Lmd+1}w
	\end{pmatrix},
\end{equation}
and then we have the following theorem.

\begin{thm}
  The $(2,1)$-type RR2T \eqref{GAL-1}--\eqref{GAL-6}
  can also be represented as
  \begin{equation}\label{HamFormalism-2}
    \left(
      k+\mu_i+\frac12
    \right)
    \veps\pp{t^{i,k}}
    \begin{pmatrix}
      u \\
      v \\
      w
    \end{pmatrix}
  =
    \mcalP_2
    \begin{pmatrix}
      \delta H_{i,k-1}/\delta u \\[2pt]
      \delta H_{i,k-1}/\delta v \\[2pt]
      \delta H_{i,k-1}/\delta w
    \end{pmatrix}
  \end{equation}
for $i\in\{2,3\}, k\geq 0$ or $i=0, k\leq 0$,
where $H_{i,k}$ are given in Theorem \ref{thm:HamFormalism-1},
and the constants
\[
  \mu_0=-\frac12,\qquad \mu_2=0,\qquad \mu_3=\frac12.
\]
\end{thm}

\begin{proof}
We continue to use the notations introduced in \eqref{def-of-coef-ab}.
Note that $\tilde L^{k+1}=A(L^kB^{-1})$ and
$L^{k+1}=(L^kB^{-1})A$ imply the recursion relations
\begin{align}
  \tilde a^{k+1}_0
&=\,
  (b^{k}_{-2})^{++} + u (b^k_{-1})^+ + vb^k_0, \label{coef-recur-3}
\\
   a^{k+1}_0
&=\,
  b^{k}_{-2} + u^- b^k_{-1} + vb^k_0 \label{coef-recur-4}
\end{align}
for each $k\geq 0$.
Subtracting \eqref{coef-recur-4} from \eqref{coef-recur-3},
and using \eqref{coef-relation-ab}, \eqref{var-derivative},
we obtain
\begin{equation}\label{Dif-b-1}
  (\Lmd^2-1)b^k_{-2}
=
  k!\left(
    (\Lmd-1)
    \left(
      w\frac{\delta H_{3,k-1}}{\delta w}
    \right)
   -(1-\Lmd^{-1})
    \left(
      u\frac{\delta H_{3,k-1}}{\delta u}
    \right)
  \right),
\end{equation}
and then applying $\frac{1}{\Lmd+1}$ to both sides of the above equation, we arrive at
\begin{equation}\label{Dif-b-2}
  (\Lmd-1)b^k_{-2}
=
  k!\left(
    \frac{\Lmd-1}{\Lmd+1}
    \left(
      w\frac{\delta H_{3,k-1}}{\delta w}
    \right)
   -\frac{1-\Lmd^{-1}}{\Lmd+1}
    \left(
      u\frac{\delta H_{3,k-1}}{\delta u}
    \right)
  \right).
\end{equation}
By using \eqref{Dut3k}--\eqref{Dwt3k}, \eqref{coef-relation-ab}--\eqref{var-derivative}
and \eqref{coef-recur-3}--\eqref{Dif-b-2}, we obtain
\begin{align*}
\veps\pfrac{u}{t^{3,k}}
&=\,
  \frac{1}{(k+1)!}
  \left(
    u\tilde a^{k+1}_0 - u(a^{k+1}_0)^+
   +v^+\tilde a^{k+1}_1 -v a^{k+1}_1
  \right) \\
&=\,
  \frac{1}{(k+1)!}
  \left[
    u\left(
      (b^k_{-2})^{++}+u(b^k_{-1})^+ +vb^k_0
     -(b^k_{-2})^+- u(b^k_{-1})^+ +v^+(b^k_0)^+
    \right)
  \right.
  \\
&\qquad
  \left.
    +v^+\left(
      (b^k_{-1})^{++}
     +u(b^k_0)^+ + vb^k_1
    \right)
   -v\left(
      b^k_{-1} + ub^k_0 + v^+b^k_1
    \right)
  \right]
\\
&=\,
  \frac{1}{k+1}
  \left(
    (\Lmd v-v\Lmd^{-1}+u\frac{1-\Lmd}{1+\Lmd}u)
    \frac{\delta H_{3,k-1}}{\delta u}
   +(u\frac{\Lmd-1}{\Lmd+1}w)\frac{\delta H_{3,k-1}}{\delta w}
  \right),
\\
  \veps\pfrac{v}{t^{3,k}}
&=\,
  \frac{1}{(k+1)!}
  \left(
    v\tilde a^{k+1}_0 - va^{k+1}_0
  \right)
=
  \frac{1}{k+1}(v(\Lmd-1)w)\frac{\delta H_{3,k-1}}{\delta w},
\\
  \veps\pfrac{w}{t^{3,k}}
&=\,
  \frac{1}{(k+1)!}
  \left(
    w\tilde a^{k+1}_0 - w(a^{k+1}_0)^-
  \right)
  \\
&=\,
  \frac{w}{(k+1)!}
  \left(
    (b^k_{-2})^{++}+u(b^k_{-1})^+ +vb^k_0
   -(b^k_{-2})^- -u^{--}(b^k_{-1})^- -v^-(b^k_0)^-
  \right)\\
&=\,
  \frac{1}{k+1}
  \left(
    (w\frac{1-\Lmd^{-1}}{\Lmd+1}u)\frac{\delta H_{3,k-1}}{\delta u}
   +(w(1-\Lmd^{-1})v)\frac{\delta H_{3,k-1}}{\delta v}
  \right.\\
&\qquad
  \left.
   +(w\frac{\Lmd^2-\Lmd^{-1}}{\Lmd+1}w)\frac{\delta H_{3,k-1}}{\delta w}
  \right),
\end{align*}
therefore \eqref{HamFormalism-2} hold true for $i=3$ and $p\geq 0$.
The validity of \eqref{HamFormalism-2} for the $\pp{t^{2,k}}$-flows and $\pp{t^{0,-k-1}}$-flows
can be verified by the same method.
We omit the details here. Thus, the theorem is proved.
\end{proof}

We will verify that the two operators $(\mcalP_1,\mcalP_2)$ given by \eqref{HamP1}, \eqref{HamP2}
forms a bihamiltonian structure in Sect.\,\ref{section: biham structure of 2,1}.

\subsection{Relation to the RR2T of type $(1,2)$}


In this subsection, we show that under a certain Miura-type transformation, the RR2T of type $(2,1)$ is equivalent to that of type $(1,2)$ up to spatial reflection,
and their corresponding Hamiltonian structures are equivalent as well.

Consider the spectral problem
\begin{equation}\label{250906-1921}
  L\psi = \lmd\psi
\end{equation}
of the RR2T of type $(2,1)$, where $L$ is the Lax operator \eqref{Lax-L},
$\psi$ and $\lmd$ are the wave function and spectral parameter respectively.
By applying a gauge transformation
\[\psi = \rho \fai\]
with new wave function $\fai$ and
$\rho$ a function to be determined, we obtain
	\begin{align*}
		\left(1+\frac{u\rho^{+}}{v\rho }\Lmd + \frac{\rho^{++}}{v\rho }\Lmd^{2}\right)^{-1}\left(-\frac{w\rho^{-}}{v\rho }\Lmd^{-1}+\frac{1}{v}\right)\fai=\frac{1}{\lmd}\fai.
	\end{align*}
	If we choose $\rho$ such that $-\frac{w\rho^{-}}{v\rho }=1$, then spectral problem \eqref{250906-1921} is transformed to
	\[
		(\mcalL|_{\Lmd^{\pm}\mapsto \Lmd^{\mp}})\fai =\frac{1}{\lmd}\fai ,
	\]
	where $\mcalL$ is the Lax operator of the RR2T of type $(1,2)$ given by \eqref{Lax type 1,2}, with
	\begin{equation}\label{3.96}
		U=-\frac{uw^{+}}{vv^{+}},\quad V=\frac{w^{+}w^{++}}{vv^{+}v^{++}},\quad W=\frac{1}{v}.
	\end{equation}

We note that the RR2T of type $(1,2)$ consists of the following flows
\begin{align*}
  \veps\pfrac{\mcalL}{s^{0,-k-1}} &= \left[\left(\mcalL^{k+1}\right)_+, \mcalL\right], \quad
  \veps\pfrac{\mcalL}{s^{2,k}} = \left[\left(\mcalM^{k+\frac12}\right)_-, \mcalL\right], \quad
  \veps\pfrac{\mcalL}{s^{3,k}} = \left[\left(\mcalM^{k+1}\right)_-,\mcalL\right]
\end{align*}
for $k\geq 0$, where $s^{0,-k-1}, s^{2,k}, s^{3,k}$ are time variables,
$\mcalL$ is given by \eqref{Lax type 1,2}, and
\[\mcalM:=\mcalL^{-1}=
  (\Lmd+W)^{-1}
  \left(
    1+U\Lmd^{-1}+V\Lmd^{-2}
  \right).
\]
More details about the above definition are omitted here,
as they are similar to the case of type $(2,1)$ in Sect.\,\ref{subsection: defo of 2,1 type}.
For example, we have
\begin{align}
		\veps\pfrac{U}{s^{0,-1}}&=U\left(W-W^{-}+U^{-}-U^{+}\right)+V^{+}-V, \label{s 0 -1 flow}\\
        \veps\pfrac{V}{s^{0,-1}}&=V\left(W-W^{--}+U^{--}-U^{+}\right),\notag\\
		\veps\pfrac{W}{s^{0,-1}}&=W\left(U-U^{+}\right). \notag
\end{align}

We note that the flow $\pp{t^{0,-1}}$ in the $(2,1)$-type RR2T, see \eqref{negative-t-0-(-1)},
has the form
\begin{align*}
\veps\pfrac{U}{t^{0,-1}} &=\,\veps\pp{t^{0,-1}}\left(-\frac{uw^+}{vv^+}\right)
=U(W-W^++U^+-U^-)+V^--V
\end{align*}
with respect to the new coordinates $(U, V, W)$ via the relations \eqref{3.96},
which coincides with \eqref{s 0 -1 flow} after spatial reflection.
In general, it can be proved that for all function $f=f(u,v,w)$ and for all $(i,k)\in\left(\{2,3\}\times\bbZ_{\geq 0}\right)\cup\left(\{0\}\times\bbZ_{<0}\right)$,
\begin{equation}\label{250907-2216}
  \veps\pfrac{f}{t^{i,k}} = c_{i,k}\left.\left(\veps\pfrac{f}{s^{i,k}}\right)\right|_{\veps\mapsto -\veps}\quad
\end{equation}
holds true for some constant $c_{i,k}$
in the sense of \eqref{3.96},
where $\pp{t^{i,k}}$, $\pp{s^{i,k}}$ are the flows of RR2T of $(2,1)$ and $(1,2)$ type respectively.
The proof of \eqref{250907-2216} follows an argument analogous to that in Sect.\,6 of \cite{AL-triham},
and the details are therefore omitted here.

Now we proceed to study the relations of the bihamiltonian formalisms of RR2T of $(2,1)$ and $(1,2)$ types.
By a method entirely analogous to that of Sect.\,\ref{subsection: biham formalism of 2,1} and after rather tedious computations,
we obtain the Hamiltonian formalism of the $(1,2)$-type RR2T as following:
\begin{align*}
  \veps\pp{s^{i,k}}
       \begin{pmatrix}
         U \\
         V \\
         W
       \end{pmatrix}
=
  \widehat\mcalP_2
  \begin{pmatrix}
    \delta G_{i,k-1}/\delta U \\
    \delta G_{i,k-1}/\delta V \\
    \delta G_{i,k-1}/\delta W
  \end{pmatrix}
=
  \widehat\mcalP_3
  \begin{pmatrix}
    \delta G_{i,k-2}/\delta U \\
    \delta G_{i,k-2}/\delta V \\
    \delta G_{i,k-2}/\delta W
  \end{pmatrix},
\end{align*}
where the Hamiltonians $G_{i,k}$ are given by
\begin{align*}
  G_{0,-k-1}&=\, \frac{1}{k+2}\int\Res\left(\mcalL^{k+2}\right)\td x, \quad k\geq -1,\\
  G_{2,k}&=\, \frac{1}{k-\frac12}\int\Res\left(\mcalM^{k-\frac12}\right)\td x, \quad k\geq 1,\\
  G_{3,k}&=\, \frac{1}{k}\int\Res\left(\mcalM^{k}\right)\td x, \quad k\geq 0,
\end{align*}
and the operators
\[
			\widehat\mcalP_{2}=
			\begin{pmatrix}
				U\Lambda U-U\Lambda^{-1}U+V\Lambda^{-1} -\Lambda V&U(\Lambda^{2}-\Lambda^{-1})V&U(1-\Lambda^{-1})W \\[4pt]
				V(\Lambda-\Lambda^{-2})U & V(\Lambda^{2}+\Lambda-\Lambda^{-1}-\Lambda^{-2})V&V(1-\Lambda^{-2})W \\[4pt]
				W(\Lambda -1)U & W(\Lambda^{2} -1)V& 0
			\end{pmatrix},
\]
\begin{equation}\label{hat P 3}
			\widehat\mcalP_{3}=
			\begin{pmatrix}
				A_{11}& A_{12}&A_{13} \\
				A_{21}&A_{22} &A_{23}\\
				A_{31} & A_{32}&	W\left(U\Lambda^{-1}-\Lambda U\right)W
			\end{pmatrix},
		\end{equation}
		\begin{align*}
			A_{11}&=U\big ((1+\Lambda^{-1}) (U-\Lambda U\Lambda)(1+\Lambda^{-1}) +2(W\Lambda-\Lambda^{-1}W)\big)U
    \clr{-}W\Lambda V
    +V\Lambda^{-1}W
\\
			&\qquad+\left(\Lambda V-V\Lambda^{-2}\right)\left(\Lambda+1\right)U+U\left(1+\Lambda^{-1}\right)\left(\Lambda^2 V-V\Lambda^{-1}\right),\\
			A_{12}&=\left(U\left(1+\Lambda^{-1}\right)\left(U\Lambda^{-1}-\Lambda U\right)+U\left(1-\Lambda^{-1}\right)W+\Lambda V-V\Lambda^{-2}\right)(\Lambda^2+\Lambda+1)V\\
			&\qquad +U\left(1+\Lambda^{-1}\right)W\left(\Lambda^2-1\right)V,\\
			A_{13}&=U\big((\Lambda+1)(\Lambda^{-1}U\Lambda^{-1}-U) +(1-\Lambda^{-1})W \big)W+\Lambda WV-V\Lambda^{-2}W\\
			&\qquad +V\left(1-\Lambda^{-2}\right)W\left(\Lambda+1\right)U,\\
			A_{22}&=\left(V\left(1+\Lambda^{-1}+\Lambda^{-2}\right)\left(U\Lambda^{-1}-\Lambda U\right)+VW-V\Lambda^{-2}W\right)\left(\Lambda^2+\Lambda+1\right)V\\
			&\qquad +V\left(1+\Lambda^{-1}+\Lambda^{-2}\right)W\left(\Lambda^2-1\right)V,\\
			A_{23}&=V\left(\left(1-\Lambda^{-2}\right)W+\left(1+\Lambda^{-1}+\Lambda^{-2}\right)\left(U\Lambda^{-1}-\Lambda U\right)\right)W,\\
A_{21} &= -A_{12}^\dag,\quad
A_{31} = -A_{13}^\dag,\quad
A_{32} = -A_{23}^\dag.
		\end{align*}

	The Jacobian $\mcalJ=(J^\afa_\beta)=
	\left(
	\sum\limits_{k\geq 0}
	\pfrac{U^\afa}{u^{\beta,k}}\p_x^k
	\right)$ between $(U^1, U^2, U^3):=(U,V,W )$ and $(u^1, u^2, u^3):=(u,v,w)$ has the form
	\begin{equation*}
				\mcalJ =
				\begin{pmatrix}
					\frac{U}{u} & -U(\Lmd+1)W & \frac{U}{w^{+}}\Lmd  \\[4pt]
					0 & -V(1+\Lmd+\Lmd^{2})W& V(\Lmd+\Lmd^{2})\frac{1}{w}\\[4pt]
					0 & -W^{2} & 0
				\end{pmatrix},
	\end{equation*}	
under the relations \eqref{3.96}. And then we have
	\begin{align}
			\mcalJ\mcalP_2\mcalJ^\dag &= \begin{pmatrix}
				 U\Lambda^{-1} U- U\Lambda U+ V\Lambda -\Lambda^{-1} V&   U(\Lambda^{-2}-\Lambda) V &  U(1-\Lmd)  W\\[4pt]
				 V(\Lambda^{-1}-\Lambda^{2}) U & V(\Lambda^{-2}+\Lambda^{-1}-\Lambda-\Lambda^{2})  V&  V(1-\Lmd^{2})  W\\[4pt]
				  W(\Lambda^{-1} -1)  U&  W(\Lambda^{-2} -1) V & 0
			\end{pmatrix} \notag\\
			&=\
			\widehat\mcalP_2|_{\Lmd^{\pm}\mapsto\Lmd^{\mp}}. \label{P2 - hat P2}
	\end{align}
	Moreover, after very complicated but straightforward computations, we can also obtain
	\begin{align} \label{P1 - hat P3}
		\mcalJ\mcalP_1\mcalJ^\dag = \widehat\mcalP_3|_{\Lmd^{\pm}\mapsto\Lmd^{\mp}}.
	\end{align}
Therefore, the bihamiltonian structures $(\mcalP_1,\mcalP_2)$ and $(\widehat\mcalP_3, \widehat\mcalP_2)$
are equivalent under the Miura-type transformation \eqref{3.96} and the spatial reflection.
These relations are analogue to that in Sect.\,7 of \cite{AL-triham}.

\subsection{Bihamiltonian structures for the two asymmetric RR2T}	
\label{section: biham structure of 2,1}

To verify that $(\mcalP_1, \mcalP_2)$ given by \eqref{HamP1}, \eqref{HamP2} forms a bihamiltonian structure,
we recall the notion of the \textit{Schouten bracket} defined on the space of local functionals of a super manifold $\hat{M}=\Pi(T^*M)$ as the cotangent bundle of $M$ with fiber's parity reversed, where $M$ is an $n$-dimensional smooth manifold.
See \cite{Liu lecture notes, Bihamiltonian cohomologies,Jocabi structure} for more details.

Fix a local trivialization $\hat U=U\times \mathbb{R}^{0|n}$ of $\hat{M}$ with local coordinates $u^1,u^2,\dots, u^n$ on $U$,
and $\mathbb{R}^{0|n}$ possesses supervariables (dual coordinates) $\theta_{1},\theta_{2},\cdots,\theta_{n}$.
Hence $\{u^{\alpha,s},\theta_{\alpha}^s|\alpha=1,2,\cdots,n,s\geq 0\}$ constitutes a local coordinate system on the jet space $J^{\infty}(\hat{M})$.
A smooth function $f\in C^{\infty}(J^{\infty }(\hat M))$ is called a differential polynomial if it depends on the jet variables
$\{u^{\alpha,{s+1}},\theta^s_\alpha\}_{s\geq 0}$ polynomially.
Denote the completion of the ring of differential polynomials by
	\begin{align*}
		\mathcal{\hat A}=C^{\infty}( U)
 [[u^{\alpha,s+1},\theta_{\alpha}^s|\alpha=1,2,\cdots,n,s\geq 0]],
	\end{align*}
and introduce the differential operator $\partial_x\in\Der(\hat\mcalA)$ via
	$$
	\partial_x=\sum_{s\geq0}\left(u^{\alpha,s+1}\frac{\partial}{\partial u^{\alpha,s}}+\theta_{\alpha}^{s+1}\frac{\partial}{\partial \theta_{\alpha}^s}\right),
	$$
	then it is natural to define the space of local functionals of $\hat M$ as
	\begin{align*}
		\hat\mcalF:=\hat\mcalA/\p_x\hat\mcalA.
	\end{align*}
For a local functional $F=f + \p_x\hat\mcalA \in \hat\mcalF$,
we always rewrite it as $\int f \td x$,
in other words, $\int\cdot\td x$ is the canonical projection $\hat\mcalA\to\hat\mcalF$.
Moreover, we call $f\in\hat\mcalA$ the density of $F$.

There is a natural gradation, called \textit{super gradation}, on $\mathcal{\hat A}$ generated by
	\begin{align*}
		\deg \theta^s_\alpha= 1,
\quad  \deg u^{\alpha,s}=\deg f=0 
	\end{align*}
for $s\geq 0$ and $f\in C^\infty(M)$,
we then denote $\hat{\mcalA}^p$ as the subspace of all homogeneous elements in $\hat\mcalA$ of super degree $p$,
and $\hat{\mcalF}^p$ as the image of $\hat{\mcalA}^p$ under the quotient map $\int\colon \hat{\mcalA}\to\hat{\mcalF}$.
The \textit{Schouten bracket} $[\ ,\ ]\colon \hat\mcalF^p\times\hat\mcalF^q\to\hat\mcalF^{p+q-1}$ is defined by
	\begin{align}
		[F,G]=\int \left( \frac{\delta  F}{\delta \theta_\alpha} \frac{\delta G}{\delta u^\alpha}
+ (-1)^p \frac{\delta  F}{\delta u^\alpha} \frac{\delta G}{\delta \theta_\alpha} \right) \td x
	\end{align}
for all $F\in\hat{\mcalF}^p$ and $G\in\hat{\mcalF}^q$, where the variational derivatives
$\frac{\delta}{\delta\theta_\afa}, \frac{\delta}{\delta u^\afa}\colon \hat\mcalF\to\hat\mcalA$ are defined by
	\begin{align*}
		\frac{\delta {F}}{\delta u^\alpha} = \sum_{s \geq 0} (-\partial_x)^s \frac{\partial  f}{\partial u^{\alpha,s}},
		\quad
		\frac{\delta {F}}{\delta \theta_\alpha} = \sum_{s \geq 0} (-\partial_x)^s \frac{\partial f}{\partial \theta^s_\alpha}
	\end{align*}
for $F=\int f\td x\in\hat\mcalF$.
	For a skew-symmetric operator-valued matrix $\mcalQ$ of the form
	\begin{align}\label{250908-1045}
		\mathcal{Q} = (\mathcal{Q}^{\alpha\beta}) = \left( \sum_{s \geq 0} \mathcal{Q}_s^{\alpha\beta} \partial^s_x \right),
	\end{align}
    we introduce its corresponding local functional
	\begin{align*}
     \iota(\mcalQ)
:= \frac{1}{2} \int \theta_\alpha \left( \mcalQ^{\afa\beta} \theta_\beta \right) \td x \in\hat\mcalF^2,
	\end{align*}
then $\mathcal Q$ is a Hamiltonian structure if and  only if
$[\iota(\mcalQ), \iota(\mcalQ)]=0$.

To apply the above theory to discrete integrable systems,
we always regard the shift operator $\Lmd$ as formal series $\rme^{\veps\p_x}=\sum_{k\geq 0}\frac{\veps^k}{k!}\p_x^k$.
For example, the operator $\mathcal{P}_1$ appeared in \eqref{HamP1}
can be rewritten as
	\begin{align*}
		\mcalP_1 &=\,
		\begin{pmatrix}
			\rme^{\varepsilon\partial_x}-\rme^{-\varepsilon\partial_x} & u_3\rme^{-\varepsilon\partial_x}-\rme^{2\varepsilon\partial_x}u_3 &0\\[4pt]
			u_3\rme^{-2\varepsilon\partial_x}-\rme^{\varepsilon\partial_x} u_3& u_3\rme^{-\varepsilon\partial_x}u_1-u_1\rme^{\varepsilon\partial_x} u_3&(\rme^{\varepsilon\partial_x}-1) u_3\\[4pt]
			0&u_3(1-\rme^{-\varepsilon\partial_x}) & 0
		\end{pmatrix},
	\end{align*}
which is of the form \eqref{250908-1045},
here we use the notation $(u_1, u_2, u_3):=(u,v,w)$.
Then the corresponding local functionals $I_a:=\iota(\mcalP_a)$, $a=1,2$ for
operators $\mcalP_1, \mcalP_2$ in \eqref{HamP1}, \eqref{HamP2} have the form
		\begin{align*}
			I_1 &=\int(\theta_{1}\theta_{1}^{+}-u_3^{++}\theta_{1}\theta_{2}^{++}-u_3^{+}\theta_{2}\theta_{1}^{+}-u_1u_3^{+}\theta_{2}\theta_{2}^{+}+u_3\theta_{3}\theta_{2}+u_3^{+}\theta_{2}\theta_{3}^{+}
			)\td x, \\
I_2 &=
\int \bigg( u_2^{+}\theta_{1}\theta_{1}^{+}+u_2u_3^{+}\theta_{2}\theta_{3}^{+}-u_2u_3\theta_{2}\theta_{3}+u_1\theta_{1}\frac{1}{1+\Lmd}\left(u_1\theta_{1}\right)+u_1\theta_{1}\frac{\Lmd-1}{1+\Lmd^{-1}}\left(u_3\theta_{3}\right)\\
					&\quad+u_3\theta_{3}\frac{\Lmd^{2}}{\Lmd+1}\left(u_3\theta_{3}\right)\bigg)
					\td x,
		\end{align*}
here we denote $f^\pm:=\Lmd^{\pm}f$ for all $f\in\hat\mcalA$.
Then by straightforward calculation, we obtain their variational derivatives
			\begin{align*}
				\frac{\delta {I_1}}{\delta u_1}&=-u_3^{+}\theta_{2}\theta_{2}^{+}, \quad
				\frac{\delta {I_1}}{\delta u_2} = 0,\\
				\frac{\delta {I_1}}{\delta u_3}
&=\theta_{3}\theta_{2}-\theta_{2}^{-}\theta_{1}-u_1^{-}\theta_{2}^{-}\theta_{2}+\theta_{2}^{-}\theta_{3}-\theta_{1}^{--}\theta_{2},\\
				\frac{\delta {I_1}}{\delta \theta_{1}}
&=\theta_{1}^{+}-u_3^{++}\theta_{2}^{++}-\theta_{1}^{-}+u_3\theta_{2}^{-},\\
				\frac{\delta {I_1}}{\delta\theta_{2}}
&=-u_3^{+}\theta_{1}^{+}-u_1u_3^{+}\theta_{2}^{+}-u_3\theta_{3}+u_3^{+}\theta_{3}^{+}+u_1^{-}u_3\theta_{2}^{-}+u_3\theta_{1}^{--},\\
				\frac{\delta {I_1}}{\delta \theta_{3}}&=u_3\theta_{2}-u_3\theta_{2}^{-}.
			\end {align*}
					\begin{align*}
						\frac{\delta {I_2}}{\delta u_1}&=\theta_{1}\frac{1-\Lmd}{\Lmd+1}\left(u_1\theta_{1}\right)+\theta_{1}\frac{\Lmd-1}{1+\Lmd^{-1}}\left(u_3\theta_{3}\right),\quad 	
						\frac{\delta {I_2}}{\delta u_2} =-u_3\theta_{2}\theta_{3}+u_3^{+}\theta_{2}\theta_{3}^{+}+\theta_{1}^{-}\theta_{1},
\\
						\frac{\delta {I_2}}{\delta u_3}
&=
  -u_2\theta_{2}\theta_{3}
  +u_2^{-}\theta_{2}^{-}\theta_{3}
  -\theta_{3}\frac{\Lmd^{-1}-1}{1+\Lmd}\left(u_1\theta_{1}\right)
    +\theta_3\frac{\Lmd^2-\Lmd^{-1}}{1+\Lmd}
    \left(
      u_3\theta_3
    \right),
\\
	\frac{\delta {I_2}}{\delta \theta_{1}}
&=u_2^{+}\theta_{1}^{+}-u_2\theta_{1}^{-}+u_1\frac{1-\Lmd}{1+\Lmd}\left(u_1\theta_{1}\right)+u_1\frac{\Lmd-1}{1+\Lmd^{-1}}\left(u_3\theta_{3}\right),
\\
						\frac{\delta {I_2}}{\delta \theta_{2}}&=u_2u_3^{+}\theta_{3}^{+}-u_2u_3\theta_{3},
\\
			\frac{\delta {I_2}}{\delta \theta_{3}}
&=u_2u_3\theta_{2}
-u_2^{-}u_3\theta_{2}^{-}
+u_3\frac{\Lmd^2-\Lmd^{-1}}{1+\Lmd}\left(u_3\theta_3\right)
  -u_3\frac{\Lmd^{-1}-1}{1+\Lmd}\left(u_1\theta_{1}\right).
						\end {align*}

\begin{thm}
  The operators $(\mcalP_1, \mcalP_2)$ given by \eqref{HamP1}, \eqref{HamP2}
  form a bihamiltonian structure of the RR2T of type $(2,1)$.
\end{thm}

\begin{proof}
  In Sect.\,\ref{subsection: biham formalism of 2,1} we have already shown that $\mcalP_1$ is a Hamiltonian structure
  by changing coordinates.
  Therefore it remains to check that $\mcalP_2$ is a Hamiltonian structure compatible with $\mcalP_1$,
  that is equivalent to
  \[
    [I_2, I_2] = [I_1, I_2] = 0.
  \]
Using the properties
$\int (\Lmd-1)f \td x=0$ for all $f\in\hat\mcalA$ and
 $\fai\fai=0$ for all $\fai\in\hat{\mcalA}^1$,
together with straightforward calculations,
 we obtain
\begin{align*}
&
							[I_2,I_2]=2\int \left(
							\frac{\delta I_2}{\delta \theta_{1}}\frac{\delta I_2}{\delta u_1}
							+
							\frac{\delta I_2}{\delta \theta_{2}}\frac{\delta I_2}{\delta u_2}
							+
							\frac{\delta I_2}{\delta \theta_{3}}\frac{\delta I_2}{\delta u_3}\right) \td x
\\
=&2\int \bigg(\left(u_2\theta_{1}^{-}\theta_{1}-u_2^{+}\theta_{1}^{+}\theta_{1}\right)\left(\frac{\Lmd-1}{1+\Lmd}\left(u_1\theta_{1}\right)-\frac{\Lmd-1}{1+\Lmd^{-1}}\left(u_3\theta_{3}\right)\right)
\\
				&\quad +u_2\theta_{1}^{-}\theta_{1}(\Lmd-1)(u_3\theta_{3})\bigg)\td x
\\							
=&2\int\bigg((1-\Lmd)\left(u_2\theta_{1}^{-}\theta_{1}\frac{1}{1+\Lmd}\left(\Lmd(u_3\theta_{3})-u_1\theta_{1}\right)+u_2u_3^{+}\theta_{1}^{-}\theta_{1}\theta_{3}^{+}-u_2u_3\theta_{1}^{-}\theta_{1}\theta_{3}\right)
\\
		&\quad+u_2\theta_{1}^{-}\theta_{1}\left(1-\Lmd^{-1}\right)\left(u_1\theta_{1}\right)\bigg) \td x=0,
\\
								& [I_1,I_2]=\int \left(
								\frac{\delta I_1}{\delta \theta_{1}}\frac{\delta I_2}{\delta u_1}
								+
								\frac{\delta I_1}{\delta \theta_{2}}\frac{\delta I_2}{\delta u_2}
								+
								\frac{\delta I_1}{\delta \theta_{3}}\frac{\delta I_2}{\delta u_3}
								+	
								\frac{\delta I_1}{\delta u_1}\frac{\delta I_2}{\delta \theta_{1}}
								+
								\frac{\delta I_1}{\delta u_2}\frac{\delta I_2}{\delta \theta_{2}}
								+
								\frac{\delta I_1}{\delta u_3}\frac{\delta I_2}{\delta \theta_{3}}
								\right) \td x\\							
=& \int\bigg[(\Lmd-1)\bigg(u_3\theta_{1}\theta_{1}^{-}\theta_{1}^{--}-u_3^2\theta_{1}\theta_{2}^{-}\theta_{3}-u_2^{-}u_3\theta_{1}^{--}\theta_{2}\theta_{2}^{-}+u_3^{+}\theta_{1}\theta_{1}^{-}\theta_{3}^{+}-u_3\theta_{1}\theta_{1}^-\theta_{3}\\
								&\quad-u_1u_3^+\theta_{1}\theta_{1}^-\theta_{2}^++u_2u_3\theta_{1}\theta_{2}\theta_{2}^{-}+u_3u_3^-\theta_{1}\theta_{2}^{-}\theta_{3}^-+u_1^-u_3u_3^+\theta_{2}\theta_{2}^-\theta_{3}^{+}-u_1^-u_3^2\theta_{2}\theta_{2}^-\theta_{3}\\
								&\quad-\theta_{1}\theta_{1}^-\frac{1}{1+\Lmd}(u_1\theta_{1})+\theta_{1}\theta_{1}^-\frac{\Lmd}{1+\Lmd}(u_3\theta_{3})+u_1^-u_3\theta_{2}\theta_{2}^-\frac{\Lmd^{-1}-1}{1+\Lmd}(u_1\theta_{1})\bigg)\\
								&\quad+(1-\Lmd^2)\left(u_3\theta_{1}^{--}\theta_{2}\frac{\Lmd^{-1}}{1+\Lmd}(u_1\theta_{1})+u_3\theta_{1}^{--}\theta_{2}\frac{\Lmd^{-1}}{1+\Lmd}(u_3\theta_{3})\right)\bigg] \td x = 0.
\end{align*}
Therefore the theorem is proved.
\end{proof}

\begin{rmk}
Since the operators $(\widehat{\mcalP}_3, \widehat{\mcalP}_2)$ given by \eqref{hat P 3}
are related to $(\mcalP_1, \mcalP_2)$ via \eqref{P2 - hat P2}--\eqref{P1 - hat P3},
it is clear that $(\widehat{\mcalP}_3, \widehat{\mcalP}_2)$ form a local bihamiltonian structure of the RR2T of type $(1,2)$.
\end{rmk}

\subsection{Central invariants for the two asymmetric RR2T}

In \cite{Hamiltonian perturbation,Deformations},
the notion of central invariants was introduced for a semisimple bihamiltonian structure which possesses hydrodynamic limit.
According to \cite{unobstructed,Hamiltonian perturbation},
these invariants provide a complete characterization of the equivalence classes of infinitesimal deformations of a semisimple bihamiltonian structure of hydrodynamic type under Miura-type transformations.
More precisely, the central invariants take the form of a set of functions of one variable.
An important example is that of a bihamiltonian integrable hierarchy controlling a cohomological field theory associated with a semisimple Frobenius manifold; in this case, the central invariants of its bihamiltonian structure are all equal to $\frac{1}{24}$ \cite{Hamiltonian perturbation,Deformations,Central invariants}.
Such a bihamiltonian structure can be viewed as a topological deformation of its dispersionless limit \cite{Normal forms}.

 In this subsection, 
 we show that
 the central invariants of the bihamiltonian structures $(\mcalP_1,\mcalP_2)$
 of the $(2,1)$-type RR2T are all equal to $\frac{1}{24}$.
 Then from the equivalence relations \eqref{P2 - hat P2}--\eqref{P1 - hat P3}
 between the bihamiltonian structures of the RR2T of type $(2,1)$ and $(1,2)$,
 it follows immediately that 
 the central invariants of the bihamiltonian structures $(\widehat{\mcalP}_3, \widehat\mcalP_2)$ of the RR2T of type $(1,2)$ are all equal to $-\frac{1}{24}$.

 Let us recall the definition of the central invariants.
 Suppose $M$ is an $n$-dimensional smooth manifold with local coordinates $\bfu=(u^1,\dots, u^n)$,
 and $(\mathcal{Q}_{1},\mathcal{Q}_{2})$ is a bihamiltonian structure defined on $J^\infty(M)$
 such that the components of them have the forms
\begin{align*}
&
  \mcalQ_a^{\afa\beta} =
    \sum_{s\geq 0}\veps^s
      \left(
        \sum_{m=0}^{s+1}
          \mcalQ_{s,m;a}^{\afa\beta}
          \p_x^m
      \right)
\\
=&\,
\left(
  g^{\afa\beta}_a\p_x + \Gamma^{\afa\beta}_{a;\gamma} u_x^\gamma
\right)
  +\veps\left(
    H^{\afa\beta}_a\p_x^2 + \cdots
  \right)
  +\veps^2\left(
    K^{\afa\beta}_a\p_x^3+\cdots
  \right)
  +O(\veps^3), \quad a=1,2,
\end{align*}
where $\mcalQ_{s,m; a}^{\alpha \beta}$ are
homogeneous differential polynomials on the jet space of $M$ with differential degree $s+1-m$.
It is well known that the coefficients $g_a^{\alpha \beta}$ for $a=1,2$
yield two flat contravariant  metrics on $M$, and $\Gamma^{\afa\beta}_{a;\gamma}$
are the contravariant Christoffel coefficients of the Levi-Civita connections for the metrics $g_a^{\alpha \beta}$ respectively.
We also note that for each $a=1,2$, the leading coefficients $H_a^{\afa\beta}$ and $K_a^{\afa\beta}$ yield,
respectively, skew-symmetric and symmetric $(2,0)$-tensors on $M$.
The bihamiltonian structure $(\mathcal{Q}_{1},\mathcal{Q}_{2})$ is called semisimple if the characteristic polynomial
						$
						P(\lmd):=\det(g_2^{ij} - \lmd g_1^{ij})
						$
						has $n$ distinct roots $\lambda^1(\mathbf{u}), \cdots, \lambda^n(\mathbf{u}) $. These roots form 
a system of local coordinate of $M$, which are called \textit{canonical coordinates} \cite{Ferapontov}
of the semisimple bihamiltonian structures.
In the canonical coordinates $\lmd=(\lmd^1,\dots,\lmd^n)$,
the metrics $g_1, g_2$ have the following diagonal forms
\begin{equation}\label{250908-2225}				
						g_1^{ij} = f^i(\lambda)\delta^{ij}, \quad g_2^{ij} = \lambda^i f^i(\lambda)\delta^{ij},
\end{equation}
where $\delta^{ij}$ are the Kronecker delta symbols,
and $f^i$ are smooth functions on $M$.
Then, the \textit{central invariants} of the semisimple bihamiltonian structure $(\mathcal{Q}_{1},\mathcal{Q}_{2})$ are given by
						\begin{equation}\label{central invariants}
									c_i(\lambda) = \frac{1}{3(f^i)^2}
\left( K_2^{ii} - \lambda^iK_1^{ii} +
\sum_{k\neq i} \frac{( H_2^{ki} - \lambda^iH_1^{ki})^2}{f^k(\lambda^k-\lambda^i)} \right)
						\end{equation}		
with respect to the canonical coordinates $\lmd=(\lmd^1,\dots,\lmd^n)$.

\begin{thm}\label{T3.10}
The central invariants of bihamiltonian structure $(\mathcal P_{1},\mathcal P_{2})$  of the RR2T of type $(2,1)$,
see \eqref{HamP1}\eqref{HamP2},
are all equal to $\frac{1}{24}$, i.e.
							\begin{align}
								c_{1}=c_{2}=c_{3}=\frac{1}{24}.
							\end{align}
						\end{thm}	
\begin{proof}				
After the rescaling $\mcalP_a\mapsto\frac{1}{\veps}\mcalP_a$, $a=1,2$,
the contravariant metrics $\eta^{\afa\beta}:=g_1^{\afa\beta}$ and
$g^{\afa\beta}:=g_2^{\afa\beta}$ corresponding to $(\mcalP_1, \mcalP_2)$
have the local form
\begin{align}\label{flat pencil}
								\left(\eta^{\alpha\beta}\right)=
								\begin{pmatrix}
									2&-3u^3&0\\[4pt]
									-3u^3&-2u^1u^3&u^3\\[4pt]
									0&u^3&0\\
								\end{pmatrix},\quad
								\left(g^{\alpha\beta}\right)=
								\begin{pmatrix}
									2u^2-\frac{1}{2}(u^1)^2&0&\frac{1}{2}u^1u^3\\[4pt]
									0&0&u^2u^3\\[4pt]
									\frac{1}{2}u^1u^3&u^2u^3&\frac{3}{2}(u^3)^2
\end{pmatrix}
\end{align}
with respect to the local coordinates $(u^1, u^2, u^3):=(u,v,w)$,
and by further calculations,
the corresponding coefficients $H^{\afa\beta}_a$, $K^{\afa\beta}_a$ for $a=1,2$ satisfy
\begin{align*}
  \left(H_1^{\afa\beta}\right)
&=\,
  \begin{pmatrix}
    0 & -\frac 32 u^3 & 0 \\[4pt]
    \frac 32 u^3 & 0 & \frac 12 u^3 \\[4pt]
    0 & -\frac12 u^3 & 0
  \end{pmatrix},
\quad
  \left(K_1^{\afa\beta}\right)
=
  \begin{pmatrix}
    \frac 13 & -\frac 32 u^3 & 0 \\[4pt]
    -\frac 32 u^3 & -\frac13 u^1 u^3 & \frac 16 u^3 \\[4pt]
    0 & \frac16 u^3 & 0
  \end{pmatrix},
\\
  \left(H_2^{\afa\beta}\right)
&=\,
  \begin{pmatrix}
    0 & 0 & \frac12 u^1 u^3 \\[4pt]
    0 & 0 & \frac12 u^2 u^3 \\[4pt]
    -\frac12 u^1 u^3 & -\frac12 u^2 u^3 & 0
  \end{pmatrix},
\quad
  \left(K_2^{\afa\beta}\right)
=
  \begin{pmatrix}
    \frac 1{24} (u^1)^2 + \frac 13 u^2 &0 & \frac{5}{24}u^1 u^3 \\[4pt]
    0 & 0 & \frac 16 u^2 u^3 \\[4pt]
    \frac{5}{24}u^1 u^3 & \frac 16 u^2 u^3 & \frac 38(u^3)^2
  \end{pmatrix}.
\end{align*}
The canonical coordinates $(\lmd^1,\lmd^2,\lmd^3)$ of $(\mcalP_1, \mcalP_2)$
are the three distinct roots of the following polynomial $P(X)$:
\begin{align*}
  P(X) &=\, X^3
  + \frac 14\Big(
       (u^1)^2 - 18 u^1 u^3 -12 u^2 -27(u^3)^2
     \Big)X^2 \\
  &\quad
   -\left(
     (u^1)^3 u^3
    +\frac12 (u^1)^2 u^2
    -\frac 92 u^1 u^2 u^3
    -3 (u^2)^2
    \right)X
   +
    \left(
      \frac 14 (u^1)^2 (u^2)^2 - (u^2)^3
    \right),
\end{align*}
and then it can be verified directly that
the entries of the Jacobian
$J=(J^i_\afa)=\left(
  \pfrac{\lmd^i}{u^\afa}
\right)$ have the form
\[
  J^i_\afa = \frac{\tilde J^i_\afa}{A_i},
\]
where
\begin{align*}
						\tilde J^i_{1}
&=\left(9(u^3)^3-u^1(u^3)^2\right)(\lambda^i)^2-\left(9u^2(u^3)^3-6(u^1)^2(u^3)^3-2u^1u^2(u^3)^2\right)\lambda^i-u^1u^2(u^3)^2,\\
								\tilde J^i_{2}
&=\left((u^1)^2(u^3)^2-12u^2(u^3)^2-9u^1(u^3)^3\right)\lambda^i-(u^1)^2u^2(u^3)^2+6(u^2)^2(u^3)^2,\\
								\tilde J^i_{3}
&=-4u^3(\lambda^i)^3-\left((u^1)^2u^3-12u^2u^3-27u^1(u^3)^2-54(u^3)^3\right)(\lambda^i)^2 \\
&\quad -\left(27u^1u^2(u^3)^2-6(u^1)^3(u^3)^2+12(u^2)^2u^3-2(u^1)^2u^2u^3\right)\lambda^i\\
&\quad-(u^1)^2(u^2)^2u^3+4(u^2)^3u^3,
\\
								A_i&=6(u^3)^2(\lambda^i)^2+\left ((u^1u^3)^2-12u^2(u^3)^2-18u^1(u^3)^3-27(u^3)^4\right)\lambda^i \\
								&\quad +9u^1u^2(u^3)^3-2(u^1)^3(u^3)^3+6(u^2)^2(u^3)^2-(u^1)^2u^2(u^3)^2,
\end{align*}
for all $ 1 \leq i \leq 3 $.	
Then the diagonal coefficients $f^i$ \eqref{250908-2225} can be computed via
							\begin{align*}
								\begin{pmatrix}
									f^1&0&0\\
									0&f^2&0\\
									0&0&f^3
								\end{pmatrix}=
								J
								\begin{pmatrix}
									2&-3u^3&0\\[4pt]
									-3u^3&-2u^1u^3&u^3\\[4pt]
									0&u^3&0\\
								\end{pmatrix}
								J^\rmT.
							\end{align*}
Since $H_a^{\afa\beta}$ and $K_a^{\afa\beta}$
are coefficients of certain $(2,0)$-tensor fields,
we can find its coefficients $H_a^{ij}, K_a^{ij}$ with respect to the canonical coordinates by using Jacobian $J=(J^i_\afa)$.
According to the definition of central invariants \eqref{central invariants},
we calculate the central invariants of the bihamiltonian structure $(\mathcal{P}_{1},\mathcal{P}_{2})$ directly and obtain
							$$
							c_{1}=c_{2}=c_{3}=\frac{1}{24}.
							$$
							The theorem is proved.
						\end{proof}
						
In the similar way, we compute the central invariants of the bihamiltonian structure $(\mathcal{P}_{2},\mathcal{P}_{1})$ and obtain
						\begin{align}
							c_{i}=-\frac{1}{24\lambda^i}, \quad 1\leq i\leq 3.
						\end{align}

\begin{rmk}
Note that the central invariants are preserved under Miura-type transformations.
Then using relations \eqref{P2 - hat P2}--\eqref{P1 - hat P3} and taking into account the spatial reflection $\veps\mapsto -\veps$,
we immediately obtain that the central invariants of the bihamiltonian structure $(\widehat\mcalP_3, \widehat\mcalP_2)$ of the RR2T of type $(1,2)$ are
\begin{align*}
							c_1=c_2=c_3 = -\frac{1}{24},
						\end{align*}
while those of the bihamiltonian structure $(\widehat\mcalP_2, \widehat\mcalP_3)$ are
						\begin{align*}
							c_i = \frac{1}{24 \lambda^i}, \quad 1 \leq i \leq 3.
						\end{align*}
These results are analogous to Theorem 5.1 of \cite{AL-triham}.
\end{rmk}

\section{From the $(2,1)$-type RR2T to a generalized Frobenius manifold}
The notion of Frobenius manifold, introduced by Boris Dubrovin \cite{Frobenius manifold,Dubrovin1996},
provides a coordinate-free formulation of the Witten-Dijkgraaf-Verlinde-Verlinde (WDVV) equations of associativity which arise in the study of 2D topological ﬁeld theory (2D TFT) \cite{TFT-2,TFT-1}.
The geometric structures of a Frobenius manifold naturally yield a bihamiltonian structure of hydrodynamic type,
from which one can construct the associated bihamiltonian integrable hierarchy of hydrodynamic type, called the \textit{Principal Hierarchy}.
In recent years, the concept of Frobenius manifolds has been extended by dropping the flatness condition $\nabla e = 0$ of the unit vector field $e$,
where $\nabla$ is the Levi­-Civita connection with respect to the Frobenius
metric, leading to the development of generalized Frobenius manifolds \cite{Liu2024,GFM}.

In this section we will see that the dispersionless limits of the RR2T of type $(2,1)$
belong to the Principal Hierarchy of a certain generalized Frobenius manifold $M$ with non-flat unity.
Moreover, $M$ is related to the Frobenius manifolds corresponding to the bi-graded Toda hierarchy \cite{extend bi Toda} and the constrained KP hierarchy \cite{cKp} via certain generalized Legendre transformations \cite{LiuQuZhang,GLt}.

\subsection{Generalized Frobenius manifold}
Consider the Landau–Ginzburg superpotential
\begin{equation}\label{superpotential-3AL}
  \lmd(p) =
    \frac{p(p^2+up+v)}{p-w}
\end{equation}
associated to the Lax operator \eqref{Lax-L}, then
$\lmd(p)$ yields a generalized Frobenius structure $M(\eta,c,e,E)$ on the $(u,v,w)$-space
in a manner similar to that described in \cite{Brini 2012}.
Precisely speaking, the Frobenius metric $\eta$ and Frobenius multiplication $(0,3)$-tensor $c$
are given by the following Landau–Ginzburg formulae \cite{Dubrovin1996}
\begin{align}
  \eta(X_1, X_2) &=\,
  \sum_{p_0\in C_\lmd}
  \Res_{p=p_0}
  \left(
    \frac{X_1(\lmd(p))\cdot X_2(\lmd(p))}{\lmd'(p)\frac{\td p}{\td\omg}}
    \td\omg
  \right),
\label{Frob-eta}\\
  c(X_1, X_2, X_3) &=\,
  \sum_{p_0\in C_\lmd}
  \Res_{p=p_0}
  \left(
    \frac{X_1(\lmd(p))\cdot X_2(\lmd(p))\cdot X_3(\lmd(p))}{\lmd'(p)\frac{\td p}{\td\omg}}
    \td\omg
  \right),  \label{Frob-c}
\end{align}
where $X_1, X_2, X_3$ are vector fields on $M$,
$C_\lmd=\Bigset{p}{\lmd'(p)=0}$ is the set of critical points of $\lmd(p)$,
and the primary differential $\td\omg$ is chosen as
\begin{equation}
  \td\omg = \frac{\td p}{p}.
\end{equation}

						Introducing new coordinates $(z^1,z^2,z^3)$ such that
						\begin{align}\label{coordinate change}
							\left\{
							\begin{aligned}
								z^1 &= - \underset{p=\infty}{\mathrm{Res}} \lambda(p) \frac{\mathrm{d}p}{p} = v + uw + w^2, \\
								z^2 &= - \sqrt{2} \underset{p=\infty}{\mathrm{Res}} \lambda^{\frac{1}{2}}(p) \frac{\mathrm{d}p}{p} = \frac{u + w}{\sqrt{2}}, \\
								z^3 &= \log w,
							\end{aligned}
							\right.
						\end{align}
then the coefficients
\[\eta_{\afa\beta}=\eta\left(\pp{z^\afa}, \pp{z^\beta}\right),\qquad
c_{\afa\beta\gamma}=c\left(\pp{z^\afa}, \pp{z^\beta}, \pp{z^\gamma}\right)\]
and $(\eta^{\afa\beta})=(\eta_{\afa\beta})^{-1}$
of the tensors \eqref{Frob-eta}--\eqref{Frob-c} have the forms
\begin{equation}
  (\eta_{\afa\beta})=(\eta^{\afa\beta}) =
  \begin{pmatrix}
      &   & 1 \\
      & 1 &   \\
    1 &   &
  \end{pmatrix},
\end{equation}
\begin{equation}\label{c-ijk}
\left\{
\begin{split}
c_{111}&=\,\frac{1}{z^1},\quad
c_{112}=0,\quad
c_{113}=1,\quad
c_{122}=1,\\
c_{123}&=\,\sqrt{2}\,\rme^{z^3},\quad
c_{133}=2\rme^{2z^3}+\sqrt{2}\,z^2\rme^{z^3},\quad
c_{222}=-z^2,\\
c_{223}&=\,0,\quad
c_{233}=\sqrt{2}\,z^1\rme^{z^3},\quad
c_{333}=4z^1\rme^{2z^3}+\sqrt{2}\,z^1z^2\rme^{z^3},
\end{split}\right.
\end{equation}
therefore the metric $\eta$ is flat, and $(z^1, z^2, z^3)$ forms a family of flat coordinates.
The potential $F$ of this Frobenius manifold is
\begin{equation}\label{F-GAL}
  F=\frac12(z^1)^2z^3
   +\frac12 z^1 (z^2)^2
   +\sqrt{2}\,z^1 z^2 \rme^{z^3}
   +\frac12 z^1 \rme^{2z^3}
   -\frac1{24} (z^2)^4
   +\frac12 (z^1)^2\log z^1,
\end{equation}
which satisfies $c_{\afa\beta\gamma}=\frac{\p^3 F}{\p z^\alpha \p z^{\beta}\p z^{\gamma}}$
and the following quasi-homogeneity condition
\begin{equation}
  \mcalL_EF=(3-d)F + \text{quadratic terms},
\end{equation}
where the Euler vector field $E$ and the charge $d$ have the form
\begin{equation}\label{Euler-E}
  E=z^1\pp{z^1}+\frac12 z^2\pp{z^2} + \frac12\pp{z^3},
\qquad
  d=1.
\end{equation}
Moreover, the unit vector field $e$ with respect to the Frobenius multiplication is
\begin{equation} \label{unity e}
  e
=
  \frac{
    z^1\p_{1}
   +\sqrt{2}\,\rme^{z^3}\p{_2}
   -\p_{3}
  }{z^1-\sqrt{2}\,z^2\rme^{z^3}}, \quad \p_\afa:=\pp{z^\afa},
\end{equation}
which can be represented as a gradient field $e=\grad_{\eta}\theta_{0,0}$, where
\begin{equation}\label{theta00}
\theta_{0,0} = z^3-\log\left(
  z^1-\sqrt{2}\,z^2\rme^{z^3}
\right) = \log\frac wv.
\end{equation}

It can be verified that the above-mentioned $M(\eta,c,e,E)$ satisfies all the axioms
of the (Dubrovin-)Frobenius manifold (see \cite{Dubrovin1996} for details), except the flatness of the unity $e$.
Therefore $M(\eta,c,e,E)$ forms a generalized Frobenius manifold in the sense of \cite{Liu2024,GFM}.

						\begin{rmk}\label{rmk:4.1}
							The flat coordinates $(z^1, z^2, z^3)$ in \eqref{coordinate change} coincide with the dispersionless limits $\varepsilon \to 0$ of \eqref{full-genera-z}. In this section we tolerate this abuse of notation for simplicity.
						\end{rmk}

We also note that $\theta_{0,0}$ in \eqref{theta00} coincides with the (dispersionless limit of the) density
of the Hamiltonian $H_{0,-1}$ in \eqref{H-0,-1}.
Moreover, the \textit{intersection form} \cite{Dubrovin1996,GFM}
\begin{equation}
  g=g^{\afa\beta}\pp{z^\afa}\otimes\pp{z^\beta},\qquad
  g^{\afa\beta}=E^\gamma c^{\afa\beta}_\gamma
\end{equation}
of this generalized Frobenius manifold can be represented as
\begin{equation}
  (g^{\afa\beta})
=
\begin{pmatrix}
   4z^1\rme^{2z^3} + 2\sqrt{2}\,z^1z^2\rme^{z^3}
  & \frac{3}{\sqrt{2}}z^1\rme^{z^3}
  & z^1+\sqrt{2}\,z^2\rme^{z^3}+\rme^{2z^3} \\[4pt]
   \frac{3}{\sqrt{2}}z^1\rme^{z^3}
  & z^1-\frac12(z^2)^2 & \frac{\rme^{z^3}}{\sqrt{2}}+\frac12z^2 \\[4pt]
  z^1+\sqrt{2}\,z^2\rme^{z^3}+\rme^{2z^3}
  & \frac{\rme^{z^3}}{\sqrt{2}}+\frac12z^2 & \frac 32
\end{pmatrix},
\end{equation}
and the contravariant metrics $\eta^{\afa\beta}, g^{\afa\beta}$ form a flat pencil on $M$.
In terms of the old coordinates $(u^1, u^2, u^3):=(u,v,w)$, the coefficient matrices of this flat pencil
coincides with \eqref{flat pencil}, which are also the leading coefficients of
the Hamiltonian operators $\mcalP_1, \mcalP_2$ in \eqref{HamP1}, \eqref{HamP2}.

\subsection{Principal Hierarchy}
In this subsection, we are to construct the Principal Hierarchy of the generalized Frobenius manifold $M$ \eqref{F-GAL} following the definitions of \cite{GFM}, and show that the dispersionless limits of the first flows of the RR2T of type $(2,1)$ belong to this Principal Hierarchy of $M$.

In general, let $M(\eta,c,e,E)$ be an $n$-dimensional generalized Frobenius manifold with non-flat unity, fix a family of flat coordinates $(z^1,\ldots,z^n)$. As it was introduced in  \cite{GFM}, the \textit{Principal Hierarchy} of $M$ is a family of evolutionary PDEs of the form
			\begin{equation}\label{PH}
				\frac{\partial z^\alpha}{\partial t^{i,k}} = \eta^{\alpha\beta} \partial_x \frac{\partial \theta_{i,k+1}}{\partial z^\beta}, \quad (i,k) \in \mathcal{I}
			\end{equation}
						about time variables $\mathbf{t} = \{t^{i,k}\}_{(i,k) \in \mathcal{I}}$, where the index set
						\begin{align*}
							\mathcal{I} = \left( \{1,2,\ldots,n\} \times \mathbb{Z}_{\geq 0} \right) \cup \left( \{0\} \times \mathbb{Z} \right),
						\end{align*}
						and the functions $\theta_{i,k}$ on $M$ satisfy the following initial conditions, recursion relations and quasi-homogeneity conditions
						\begin{equation}\label{4.16}
							\theta_{\alpha,0} = \eta_{\alpha\beta} z^\beta, \quad \nabla \theta_{0,0} = e,
						\end{equation}
						\begin{equation}\label{4.17}
							\partial_\alpha \partial_\beta \theta_{i,k+1} = c_{\alpha\beta}^\gamma \partial_\gamma \theta_{i,k},\quad c_{\alpha\beta}^\gamma=\eta^{\gamma\lambda }c_{\lambda\alpha\beta},
						\end{equation}
						\begin{equation}\label{4.18}
							\mathcal{L}_E \theta_{i,k} = \left( k + \mu_i + \mu_0 + 1 \right) \theta_{i,k} + \sum_{s=1}^k \theta_{\alpha,k-s} \left( \widetilde{R}_s \right)_i^\alpha + \left( \widetilde{R}_{k+1} \right)_i^{n+1},
						\end{equation}
						for $(i,k) \in \mathcal{I}$, here all the lower Greek indices are assumed to range over $\{1,2,\ldots,n\}$. We further recall that $\widetilde \mu_i$ and $(\widetilde{R}_s)^i_j$ are entries of certain constant $(n+2) \times (n+2)$ matrices
						$$
						\widetilde{\mu} = \text{diag}\left(\mu_0; \mu_1, \ldots, \mu_n; -\mu_0\right), \quad \widetilde{R} = \sum_{s\geq 1} \widetilde{R}_s,
						$$ whose rows and columns are labeled by indices in $\{0; 1, \ldots, n; n+1\}$. More details about the properties of $\widetilde{\mu}, \widetilde{R}$ can be found in Remark 5.2 of \cite{GFM} and Remark 3.4 of \cite{LiuQuZhang}.
The flows of the Principal Hierarchy \eqref{PH} can be rewritten as the following Hamiltonian system of hydrodynamic type
						\begin{equation}
							\frac{\partial z^\alpha}{\partial t^{i,k}} = (\mathcal{P}_1^{[0]})^{\alpha\beta} \frac{\delta H_{i,k}^{[0]}}{\delta z^\beta}, \quad (i,k) \in \mathcal{I},
						\end{equation}
						and satisfy the following bihamiltonian recursion relations
						\begin{equation}
							(\mathcal{P}_2^{[0]})^{\alpha\beta} \frac{\delta H_{i,k-1}^{[0]}}{\delta z^\beta}
= \left(k + \mu_i + \frac{1}{2}\right) (\mathcal{P}_1^{[0]})^{\alpha\beta} \frac{\delta H_{i,k}^{[0]}}{\delta z^\beta}
+ \sum_{s=1}^k \left(\widetilde{R}_s\right)^{\gamma}_{i} (\mathcal{P}_1^{[0]})^{\alpha\beta} \frac{\delta H_{\gamma,k-s}^{[0]}}{\delta z^\beta},
						\end{equation}
						where the compatible Hamiltonian operators $\{\mathcal{P}_1^{[0]}, \mathcal{P}_2^{[0]}\}$ have the form
						$$
						(\mathcal{P}_1^{[0]})^{\alpha\beta} = \eta^{\alpha\beta} \partial_x, \quad (\mathcal{P}_2^{[0]})^{\alpha\beta} = g^{\alpha\beta} \partial_x + \Gamma^{\alpha\beta}_{\gamma} z_x^\gamma,
						$$
						and the Hamiltonians
						\begin{align*}
							H_{i,k-1}^{[0]} = \int \theta_{i,k}(\mathbf{z}(x)) \, \td x, \quad (i,k) \in \mathcal{I}.
						\end{align*}
						Here $g^{\alpha\beta} = E^\gamma c_\gamma^{\alpha\beta}$ and $\Gamma_\gamma^{\alpha\beta} = \left(\frac{1}{2} - \mu_\beta\right) c_\gamma^{\alpha\beta}$ are the intersection form and the contravariant coefficients of the Levi-Civita connection of $(g^{\alpha\beta})$. More details can be found in \cite{Normal forms,GFM}.

In the special case of the 3-dimensional generalized Frobenius manifold $M$ \eqref{F-GAL}, the first few terms of $\{\theta_{i,k}\}$ can be obtained by solving \eqref{4.16}--\eqref{4.18} directly:
						\begin{align}
							\theta_{1,0} &= z^3, \quad
							\theta_{1,1} = z^1 \left( z^3 + \log z^1 - 1 \right) + \frac{1}{2} (z^2)^2 + \sqrt{2} z^2 \rme^{z^3} + \frac{1}{2} \rme^{2z^3}, \label{1,1}\\
							\theta_{2,0} &= z^2, \quad
							\theta_{2,1} = z^1 z^2 + \sqrt{2} z^1 \rme^{z^3} - \frac{1}{6} (z^2)^3, \\
							\theta_{3,0} &= z^1, \quad
							\theta_{3,1} = \frac{1}{2} (z^1)^2 + \sqrt{2} z^1 z^2 \rme^{z^3} + z^1 \rme^{2z^3},  \\
							\theta_{0,-1} &= \frac{-z^1 + \rme^{2z^3}}{(z^1 - \sqrt{2} z^2 \rme^{z^3})^2}, \quad
							\theta_{0,1} = z^1 z^3 + \frac{1}{2} (z^2)^2,
						\end{align}
and $\theta_{0,0}$ is given by \eqref{theta00}.
The data $\tilde{\mu}$, $\tilde{R}$ have the form
						\begin{equation}\label{data}
							\tilde{\mu} = \begin{pmatrix}
								-\frac{1}{2} & & & &\\[4pt]
								& -\frac{1}{2} & & &\\[4pt]
								& & 0 & & \\[4pt]
								& & & \frac{1}{2} &\\[4pt]
								& & & &\frac{1}{2}
							\end{pmatrix}, \qquad
							\tilde{R} = \tilde{R}_1 = \begin{pmatrix}
								0 & 0 & 0 & 0 & 0 \\[4pt]
								0 & 0 & 0 & 0 & 0 \\[4pt]
								0 & 0 & 0 & 0 & 0 \\[4pt]
								\frac{1}{2} & \frac{3}{2} & 0 & 0 & 0 \\[4pt]
								-\frac{1}{2} & \frac{1}{2} & 0 & 0 & 0
							\end{pmatrix},
						\end{equation}
whose rows and columns are labeled by $\{0,1,2,3,4\}$. Then by straightforward calculation, 
the first few flows of the Principal Hierarchy \eqref{PH} of $M$ have the form
						\begin{align}
							\frac{\partial}{\partial t^{1,0}} \begin{pmatrix} u \\[4pt] v \\[4pt] w \end{pmatrix}
&= \frac{1}{v + uw + w^2}
							\begin{pmatrix}
								v + uw & -w & uw + 2v \\[4pt]
								vw & v & uv + 2vw \\[4pt]
								w^2 & w & v + 2uw + 3w^2
							\end{pmatrix}
							\begin{pmatrix} u_x \\[4pt] v_x \\[4pt] w_x \end{pmatrix},  \label{PH t10}
\\
							\frac{\partial}{\partial t^{2,0}} \begin{pmatrix} u \\[4pt] v \\[4pt] w \end{pmatrix} &= \frac{1}{\sqrt{2}} \begin{pmatrix} -u & 2 & u \\[4pt] 0 & 0 & 2v \\[4pt] w & 0 & 3w \end{pmatrix} \begin{pmatrix} u_x \\[4pt] v_x \\[4pt] w_x \end{pmatrix},
\label{250909-PH-1}\\
							\frac{\partial}{\partial t^{3,0}}
\begin{pmatrix} u \\[4pt] v \\[4pt] w \end{pmatrix} &=
\begin{pmatrix}
  0 & w & 2v + 2uw \\[4pt]
  vw & 0 & uv + 4vw \\[4pt]
  2w^2 & w & 2uw + 6w^2
\end{pmatrix} \begin{pmatrix} u_x \\[4pt] v_x \\[4pt] w_x \end{pmatrix},
\label{250909-PH-2}
\\
							\frac{\partial}{\partial t^{0,-1}} \begin{pmatrix} u \\[4pt] v \\[4pt] w \end{pmatrix} &=
\frac{1}{v^2} \begin{pmatrix} 0 & -3w & 2v \\[4pt] vw & -2uw &  uv
\\[4pt] 0 & w & 0 \end{pmatrix} \begin{pmatrix} u_x \\[4pt] v_x \\[4pt] w_x \end{pmatrix}
\label{250909-PH-3}
						\end{align}
with respect to the change of coordinates \eqref{coordinate change}.
Then it is straightforward to verify that the flows $\pp{t^{0,-1}}, \pp{t^{2,0}}$ and $\pp{t^{3,0}}$
in \eqref{250909-PH-1}--\eqref{250909-PH-3}
coincide with that in the dispersionless limits of the flows of the $(2,1)$-type RR2T, see \eqref{rad-flow-u}--\eqref{negative-t-0-(-1)}.
						
In fact, all the densities $\{\theta_{i,k}\}_{(i,k)\in\mathcal{I}}$ of $M$ \eqref{F-GAL} can be represented in terms of the superpotential \eqref{superpotential-3AL} in a manner similar to that described in Sect.\,2 of \cite{extend AL}.
To see this, first note that the superpotential \eqref{superpotential-3AL} has the form
						\begin{align*}
							\lambda = p^2 \left( 1 + \frac{u}{p} + \frac{v}{p^2} \right) \left( 1 - \frac{w}{p} \right)^{-1},
						\end{align*}
and then we define
						\begin{equation*}
							\begin{split}
								\lambda^{\frac{1}{2}} &= p \left( 1 + \frac{u}{p} + \frac{v}{p^2} \right)^{\frac{1}{2}} \left( 1 - \frac{w}{p} \right)^{-\frac{1}{2}} \\
								&= p + \frac{1}{\sqrt{2}} z^2 + \left( \frac{1}{2} z^1 - \frac{1}{4} (z^2)^2 \right) \frac{1}{p}
+ \left( \frac{1}{2} z^1 \rme^{z^3} - \frac{1}{2\sqrt{2}} z^1 z^2 + \frac{1}{4\sqrt{2}} (z^2)^3 \right) \frac{1}{p^2} + \cdots,
							\end{split}
						\end{equation*}
						which is a formal Laurent series at $ p = \infty $. Denote formal variable $ q $ as
						\begin{equation}
							q := p - w = p - \rme^{z^3},
						\end{equation}
						and notice that the superpotential $ \lambda $ can be expressed by the following three ways
						\begin{align*}
							\lambda &= \frac{z^1 \rme^{z^3}}{q} \left( 1 + \rme^{-z^3} q \right) \left( 1 + \frac{\sqrt{2} z^2 + \rme^{z^3}}{z^1} q + \frac{1}{z^1} q^2 \right)\\
							&= q^2 \left( 1 + \frac{\rme^{z^3}}{q} \right) \left( 1 + \frac{\sqrt{2} z^2 + \rme^{z^3}}{q} + \frac{z^1}{q^2} \right)\\
							&= z^1 \left( 1 + \frac{\rme^{z^3}}{q} \right) \left( 1 + \frac{\sqrt{2} z^2 + \rme^{z^3}}{z^1} q + \frac{1}{z^1} q^2 \right),
						\end{align*}
						then we introduce the following formal series of $ q = p - \rme^{z^3} $
						\begin{align}
							\log_+ \lambda &:= -\log q + (z^3 + \log z^1)
+ \sum_{s=1}^{\infty} \frac{(-1)^{s+1}}{s} \left( \Big( \rme^{-z^3} q \Big)^s +
\Big( \frac{\sqrt{2} z^2 + \rme^{z^3}}{z^1} q + \frac{1}{z^1} q^2 \Big)^s \right),
							\nonumber\\
							\log_- \lambda &:= 2 \log q + \sum_{s=1}^{\infty} \frac{(-1)^{s+1}}{s}
\left( \Big( \frac{\rme^{z^3}}{q} \Big)^s + \Big( \frac{\sqrt{2} z^2 + \rme^{z^3}}{q} + \frac{z^1}{q^2} \Big)^s \right),
							\nonumber\\
							\widetilde{\log} \lambda &:= \log z^1 + \sum_{s=1}^{\infty} \frac{(-1)^{s+1}}{s}
\left( \Big( \frac{\rme^{z^3}}{q} \Big)^s +
\Big( \frac{\sqrt{2} z^2 + \rme^{z^3}}{z^1} q + \frac{1}{z^1} q^2 \Big)^s \right),
						\end{align}
						\begin{equation}
							\log \lambda := \frac{2 \log_+ \lambda + \log_- \lambda}{3}.
						\end{equation}
						\begin{prop}\label{prop-250909}
							The densities $\{\theta_{i,k}\}_{(i,k)\in\mathcal{I}}$ of $M$ \eqref{F-GAL} can be chosen as
							\begin{align}
								\theta_{2,k} &= -\frac{2^{k+\frac{1}{2}}}{(2k+1)!!}\res_{p=\infty}\left(\lambda^{k+\frac{1}{2}}\frac{\mathrm{d}p}{p}\right),\label{2,k} \\
								\theta_{3,k} &= -\frac{1}{(k+1)!}\res_{p=\infty}\left(\lambda^{k+1}\frac{\mathrm{d}p}{p}\right), \label{3,k}\\
								\theta_{0,-k-1} &= (-1)^{k+1}k! \res_{p=0}\left(\lambda^{-k-1}\frac{\mathrm{d}p}{p}\right), \label{0,-k-1}
							\end{align}
							for each $k\geq 0$, and
							\begin{align}
								\theta_{1,k} &= \frac{3}{2k!}\underset{{p=w}}{\Res}\left(\lambda^k\left(\log\lambda-H_k\right)\frac{\mathrm{d}p}{p}\right), \\
								\theta_{0,k} &= \theta_{1,k}-\frac{1}{k!}\underset{{p=0}}{\Res}
\left(\lambda^k\left(\widetilde\log\lambda-H_k\right)\frac{\td p}{p}\right), \label{0,k}
							\end{align}
							for each $k\geq 1$, where $H_k = 1 + \frac{1}{2} + \cdots + \frac{1}{k}$.
						\end{prop}
						\begin{proof}
							Denote $\Theta_{2,k}$, $\Theta_{3,k}$, $\Theta_{0,-k-1}$, $\Theta_{1,k}$, and $\Theta_{0,k}$ as the right-hand side of \eqref{2,k}--\eqref{0,k} respectively, then
							\begin{align*}
								\Theta_{1,1} &= \frac{3}{2}\underset{{p=w}}{\Res}\left(\frac{\lambda}{p}(\log\lambda - 1)\,\mathrm{d}p\right) = \frac{3}{2}\underset{{q=0}}{\Res}\left(\frac{\lambda}{q + \rme^{z^3}}(\log\lambda - 1)\,\mathrm{d}q\right)\\
								&=\frac{3}{2}\underset{{q=0}}{\Res}\left(\left(q+(\sqrt{2}z^2+\rme^{z^3})+\frac{z^1}{q}\right)(\log \lambda-1) \mathrm{d}q\right)\\
								&=z^1(z^3+\log z^1-1)+\frac{1}{2}(z^2)^2+\sqrt{2}z^2\rme^{z^3}+\frac{1}{2}\rme^{2z^3},
							\end{align*}
							which coincides with the $\theta_{1,1}$ given in \eqref{1,1}. In the similar way, we can also show
							\begin{align*}
								\Theta_{2,0}=\theta_{2,0},\quad 	\Theta_{3,0}=\theta_{3,0},\quad 	\Theta_{0,-1}=\theta_{0,-1},\quad
								\Theta_{0,1}=\theta_{0,1},
							\end{align*}
							directly. To prove the validity of formulae \eqref{2,k}--\eqref{0,k}, we need to show that $\Theta_{i,k}$ satisfies the  relations \eqref{4.17}--\eqref{4.18}, i.e.
							\begin{equation}\label{4.40}
								\partial_\alpha \partial_\beta \Theta_{i,k+1} = c^\gamma_{\alpha\beta} \partial_\gamma \Theta_{i,k}, \qquad \partial_\alpha := \frac{\partial}{\partial z^\alpha},
							\end{equation}
							where the Frobenius coefficients $c_{\alpha\beta\gamma}$ of $M$ are given in \eqref{c-ijk}, $c^\gamma_{\alpha\beta} = \eta^{\gamma\delta} c_{\alpha\beta\delta}$, and
							\begin{align}\label{4.41}
								\mathcal{L}_E \Theta_{2,k} &= \left(k + \frac{1}{2}\right) \Theta_{2,k}, \quad \mathcal{L}_E\Theta_{3,k} = (k + 1) \Theta_{3,k}, \quad  \mathcal{L}_E\Theta_{0,-k-1} = (-k - 1) \Theta_{0,-k-1}, \nonumber \\
								\mathcal{L}_E\Theta_{1,\ell} &= \ell \Theta_{1,\ell} + \frac{3}{2} \Theta_{3,\ell-1}, \quad \mathcal{L}_E \Theta_{0,\ell}= \ell \Theta_{0,\ell} + \frac{1}{2} \Theta_{3,\ell-1},
							\end{align}
							for each $k \geq 0$ and $\ell \geq 1$, where the Euler vector field $E$ is given in \eqref{Euler-E}.
Notice that
							$$
							\partial_\alpha \partial_\beta \Theta_{1,k+1} = \frac{3}{2k!} \underset{{p=w}}{\Res} \left[ \left( \frac{\partial^2 \lambda}{\partial z^\alpha \partial z^\beta} \cdot \lambda^k (\log \lambda - H_k) + \frac{\partial \lambda}{\partial z^\alpha} \frac{\partial \lambda}{\partial z^\beta} \cdot k \lambda^{k-1} (\log \lambda - H_{k-1}) \right) \frac{\mathrm{d}p}{p} \right],
							$$
							and then by straightforward calculation we obtain
							\begin{align*}
								\partial_1 \partial_3 \Theta_{1,k+1} - c^\alpha_{13} \partial_\alpha \Theta_{1,k} &= \frac{3}{2k!} \underset{{p=w}}{\Res} \left[ \frac{\partial}{\partial p} \left( \frac{-\rme^{z^3}}{p - \rme^{z^3}} \cdot \lambda^k (\log \lambda - H_k) \right) \mathrm{d}p \right] = 0, \\
								\partial_3 \partial_3 \Theta_{1,k+1} - c^\alpha_{33} \partial_\alpha \Theta_{1,k} &= \frac{3}{2k!} \underset{{p=w}}{\Res} \left[ \frac{\partial}{\partial p} \left( \frac{-z^1 \rme^{z^3} p}{(p - \rme^{z^3})^2} \cdot \lambda^k (\log \lambda - H_k) \right) \mathrm{d}p \right] = 0.
							\end{align*}
The validity of other cases of \eqref{4.40} can be verified in a similar way, and we omit its details here.
Finally, by the residue theorem, $\Theta_{3,k}$ \eqref{3,k} can also be represented as
							$$
							\Theta_{3,k} = \frac{1}{(k+1)!} \underset{{p=w}}{\Res} \left( \lambda^{k+1} \frac{\mathrm{d}p}{p} \right),
							$$
							from which the validity of the quasi-homogeneous conditions \eqref{4.41} can be verified by straightforward calculations.
The proposition is proved.
						\end{proof}
						\begin{rmk}
It is clear that $\theta_{2,k}$, $\theta_{3,k}$, $\theta_{0,-k-1}$ in \eqref{2,k}–\eqref{0,-k-1} coincide with the leading terms
							of the Hamiltonian densities of $H_{2,k-1}$, $H_{3,k-1}$, $H_{0,-k-2}$ in \eqref{AL-H2k}--\eqref{AL-H0-k-1} respectively,
and then from Remark \ref{rmk:4.1} and the definition of Principal Hierarchy of $M$,
it is clear that the dispersionless limits of flows of the $(2,1)$-type RR2T \eqref{GAL-1}--\eqref{GAL-6} coincide with the corresponding flows in the Principal Hierarchy of $M$ \eqref{F-GAL}.
						\end{rmk}

In other words, the flows \eqref{GAL-1}--\eqref{GAL-6} are certain deformations of the $\frac{\partial}{\partial{t^{2,k}}}$,
$\frac{\partial}{\partial{t^{3,k}}}$ and $\frac{\partial}{\partial{t^{0,-k-1}}}$-flows in the Principal Hierarchy of $M$ \eqref{F-GAL}.

\begin{rmk}\label{rmk:extended flows}
Since the first bihamiltonian cohomology of the
dispersionless limit of the bihamiltonian structure \eqref{HamP1}, \eqref{HamP2} is trivial \cite{Hamiltonian perturbation, Liu lecture notes},
there exist unique deformations of the flows $\{\pp{t^{0,k}}, \pp{t^{1,k}}\}_{k\geq 0}$ in the Principal Hierarchy of $M$ \eqref{F-GAL}, such that these deformed flows commute with all the flows \eqref{GAL-1}--\eqref{GAL-6} of the $(2,1)$-type RR2T.
\end{rmk}

The hierarchy consisting of the flows \eqref{GAL-1}--\eqref{GAL-6} and the above-mentioned deformed $\{\pp{t^{0,k}}, \pp{t^{1,k}}\}_{k\geq 0}$ flows
is called the \textit{extended RR2T of type $(2,1)$},
which is analogous to the concept of the extended Ablowitz--Ladik hierarchy \cite{extend AL}.
We will provide an explicit expression of this deformed $\frac{\partial}{\partial{t^{1,0}}}$-flow in Remark \ref{rmk:250915-2012}.

\subsection{Generalized Legendre transformations}
\label{subsection:Legendre transf}

There are two specific generalized Legendre transformations that relate $M$ \eqref{F-GAL} to the Frobenius manifolds associated with the bi-graded Toda hierarchy and the constrained KP hierarchy, respectively. To see this, we first review some preliminary concepts in \cite{LiuQuZhang,GLt}.

Let $ M(\eta, c, e, E) $ be an $ n $-dimensional generalized Frobenius manifold. A vector field $ B \in \mathrm{Vect}(M) $ is called \textit{Legendre}, if
						\begin{equation*}
							X \cdot \nabla_Y B = Y \cdot \nabla_X B, \quad \forall X, Y \in \mathrm{Vect}(M),
						\end{equation*}
						where $ \nabla $ is the Levi-Civita connection with respect to $ \eta $, and $ \cdot = c $ is the Frobenius multiplication. If $ B $ is an invertible Legendre field, i.e., $ B $ is Legendre and $ \exists B^{-1} \in \mathrm{Vect}(M) $ such that $ B \cdot B^{-1} = e $, then there exists a new flat metric $ \hat{\eta} $ on $ M $ such that
						\begin{equation}
							\hat{\eta}(X, Y) = \eta(B \cdot X, B \cdot Y), \quad \forall X, Y \in \mathrm{Vect}(M).
						\end{equation}
						A Legendre field $ B $ is called \textit{quasi-homogeneous}, if there exists $ \mu_B \in \mathbb{C} $ such that
						\begin{align*}
							[E, B] = \left( \mu_B - \frac{2-d}{2} \right) B,
						\end{align*}
						where $ d $ is the charge of $ M(\eta, c, e, E) $.
						\par Suppose $ B $ is a quasi-homogeneous, invertible Legendre field on $ M $, then $ \hat{M}:= M(\hat{\eta}, c, e, E) $ is also a generalized Frobenius manifold, which is called the \textit{generalized Legendre transformation} of $ M $ along $ B $. Note that $ \hat{M} $ and $ M $ share the same Frobenius multiplication $ c $ (therefore the same unity $ e $) and Euler field $ E $. Moreover, the inversion $ \hat{B} := B^{-1} $ is quasi-homogeneous Legendre field on $ \hat{M} $, and transforms $ \hat{M} $ back to $ M $, so we also say the pairs $ (M, B) $ and $ (\hat{M}, \hat{B}) $ are \textit{Legendre dual} to each other.
						\par Fix a family of flat coordinates $ (z^1, \ldots, z^n) $ on $ M $, and a quasi-homogeneous, invertible Legendre field $ B = B^\alpha \frac{\partial}{\partial z^\alpha} $, then we can find a family of flat coordinates $ (\hat{z}^1, \ldots, \hat{z}^n) $ with respect to the new metric $ \hat{\eta} $, such that
						\begin{equation}\label{relations4.43}
							\frac{\partial \hat{z}^\alpha}{\partial z^\beta} = B^\gamma c_{\beta \gamma}^\alpha, \quad \eta \left( \frac{\partial}{\partial z^\alpha}, \frac{\partial}{\partial z^\beta} \right) = \hat{\eta} \left( \frac{\partial}{\partial \hat{z}^\alpha}, \frac{\partial}{\partial \hat{z}^\beta} \right),
						\end{equation}
						where $ c_{\beta \gamma}^\alpha $ is the coefficients of the Frobenius multiplication $ c $ with respect to $ z^\alpha $-coordinates,
and the potential $ F $ and $ \hat{F} $ of $ M $ and $ \hat{M} $ are related by
						\begin{equation}\label{4.44}
							\frac{\partial^2 F}{\partial z^\alpha \partial z^\beta} = \frac{\partial^2 \hat{F}}{\partial \hat{z}^\alpha \partial \hat{z}^\beta}.
						\end{equation}
If both the unity $e$ and the Legendre field $B$ are flat, i.e. $\nabla e = \nabla B = 0 $, then corresponding
						generalized Legendre transformation coincide with the usual Legendre transformation
						introduced by Dubrovin \cite{Dubrovin1996}.

 Now we are to apply the above-mentioned generalized Legendre transformations to the generalized Frobenius manifold $M$
 \eqref{F-GAL} with non-flat unity $e$ \eqref{unity e}.

						\begin{ex} \label{ex: Legendre to bi-graded Toda}
							(Bi-graded Toda).
Consider the quasi-homogeneous, invertible Legendre field
							\begin{align*}
								B_1=\frac{\partial}{\partial z^1}
							\end{align*}
							on $M$  \eqref{F-GAL}, then from the relations \eqref{relations4.43} we obtain
the new flat coordinates $\hat z^\afa$ as
							\begin{equation}\label{Legendre bigToda change}
								\left\{
								\begin{aligned}
									\hat{z}^1 &= z^1 + \sqrt{2} z^2 \rme^{z^3} + \rme^{2z^3}, \\[4pt]
									\hat{z}^2 &= z^2 + \sqrt{2}\, \rme^{z^3}, \\[4pt]
									\hat{z}^3 &= z^3 + \log z^1,
								\end{aligned}
								\right.
								\quad
								\big( \hat{\eta}(\partial_{\hat{z}^\alpha}, \partial_{\hat{z}^\beta}) \big) =
								\begin{pmatrix}
									0&0 & 1 \\[4pt]
									0&1 & 0\\[4pt]
									1&0&0
								\end{pmatrix}.
							\end{equation}
							Using \eqref{4.44}, we obtain the potential $ \hat{F} $ of $ \hat{M} = M(\hat{\eta}, c, e, E) $ as follows:
							\begin{equation}\label{4.46}
								\hat{F} = \frac{1}{2} (\hat{z}^1)^2 \hat{z}^3 + \frac{1}{2} \hat{z}^1 (\hat{z}^2)^2 - \frac{1}{24} (\hat{z}^2)^4 + \sqrt{2} \hat{z}^2 \rme^{\hat{z}^3}.
							\end{equation}
							The Euler field $ E $ and the unity $e$ have the form
							\begin{equation*}
								E = \hat{z}^1 \frac{\partial}{\partial \hat{z}^1} + \frac{1}{2} \hat{z}^2 \frac{\partial}{\partial \hat{z}^2} + \frac{3}{2} \frac{\partial}{\partial \hat{z}^3}, \quad e = \frac{\partial}{\partial \hat{z}^1}
							\end{equation*}
with respect to the new coordinates  $\hat z^\afa$.
						\end{ex}

Note that the above $\hat M$ \eqref{4.46} coincides with the corresponding Frobenius manifold of
a three-component extended bi-graded Toda hierarchy \cite{extend bi Toda} with Lax operator
						\begin{align}\label{hatL}
							\hat{L}=\hat \Lmd^2+ \hat u \hat \Lmd+ \hat v+ \hat w\hat\Lmd^{-1},\quad \hat{\Lmd}=\exp\left(\varepsilon \frac{\partial}{\partial t^{1,0}}\right),
						\end{align}
						i.e. $\hat{M}$ \eqref{4.46} can be obtained by applying the Landau--Ginzburg formulae
\eqref{Frob-eta}--\eqref{Frob-c} to the superpotential and primary differential
						\begin{align*}
							\hat{\lambda}(p) = p^2 + \hat{u}p + \hat{v} + \frac{\hat{w}}{p}, \quad \mathrm{d}\omega = \frac{\mathrm{d}p}{p},
						\end{align*}
						with respect to the change of coordinates
						\begin{equation}\label{bi-graded Toda change coord}
							\hat{z}^1 = \hat{v}, \quad \hat{z}^2 = \frac{\hat{u}}{\sqrt{2}}, \quad \hat{z}^3 = \log \hat{w}.
						\end{equation}

We also remark that $\hat{M}$ \eqref{4.46} is isomorphic to the Frobenius manifold $M_{\tilde{W}^{(1)}(\mathsf{A}_2)}$, which is constructed on the orbit space of a certain extended affine Weyl group; see Example 2.2 in \cite{affine Weyl groups} for details.
						\begin{ex}
							(Constrained KP). Consider the quasi-homogeneous, invertible Legendre field
							\begin{align*}
								B_2=\frac{\partial}{\partial z^2}
							\end{align*}
							on $M$ \eqref{F-GAL}, then from the relations \eqref{relations4.43} we obtain
the new flat coordinates $\tilde z^\afa$ as
							\begin{equation}
								\left\{
								\begin{aligned}
									\tilde{z}^1 &= \sqrt{2} z^1 \rme^{z^3} , \\
									\tilde{z}^2 &= z^1 - \frac{1}{2}(z^2)^2, \\
									\tilde{z}^3 &= z^2 + \sqrt{2}\, \rme^{z^3},
								\end{aligned}
								\right.
								\quad
								\big( \tilde{\eta}(\partial_{\tilde{z}^\alpha}, \partial_{\tilde{z}^\beta}) \big) =
								\begin{pmatrix}
									0&0 & 1 \\[4pt]
									0&1 & 0\\[4pt]
									1&0&0
								\end{pmatrix}.
							\end{equation}
							Using \eqref{4.44}, we obtain the potential $ \tilde{F} $ of $ \tilde{M} = M(\tilde{\eta}, c, e, E) $ as follows:
							\begin{equation}\label{4.51}
								\tilde{F} = \tilde z^1\tilde z^2\tilde z^3+\frac{1}{6}\tilde z^1(\tilde z^3)^3+\frac{1}{6}(\tilde z^2)^3+\frac{1}{2}(\tilde z^1)^2\log {\tilde z^1}.
							\end{equation}
							The Euler field $ E $ and the unity $ e $ have the form
							\begin{equation*}
								E = \frac{3}{2}\tilde{z}^1 \frac{\partial}{\partial\tilde{z}^1} +  \tilde{z}^2 \frac{\partial}{\partial \tilde{z}^2} + \frac{1}{2} \tilde z^3\frac{\partial}{\partial \tilde{z}^3}, \quad e = \frac{\partial}{\partial \tilde{z}^2}.
							\end{equation*}
						\end{ex}
						\par Note that the above $\tilde M$ \eqref{4.51} coincides with the corresponding Frobenius manifold of the
						constrained KP hierarchy \cite{cKp} with Lax operator
						\begin{align}\label{KP-L}
							\tilde{L}=\frac{\varepsilon^2}{2}\partial_{\tilde x}^2+\tilde{z}^2+\left(\varepsilon\partial_{\tilde x}-\tilde z^3\right)^{-1}\tilde z^1,\quad \partial_ {\tilde{x}}= \frac{\partial}{\partial t^{2,0}},
						\end{align}
						i.e. $\tilde{M}$ \eqref{4.51} can be obtained by applying the Landau--Ginzburg formulae \eqref{Frob-eta}--\eqref{Frob-c}
to the superpotential and primary differential
						\begin{align*}
							\tilde{\lambda}(p) = \frac{p^2}{2} + \tilde{z}^2 + \frac{\tilde{z}^1}{p-\tilde z^3}, \quad \td \omega = \td p.
						\end{align*}

\section{Linear reciprocal transformations}

In Sect.\,\ref{subsection:Legendre transf}, we introduced two specific generalized Legendre transformations,
relating $M$ \eqref{F-GAL} to the Frobenius manifolds corresponding to
the bi-graded Toda hierarchy and the constrained KP hierarchy respectively.
Therefore, from Theorem 3.13 and Theorem 4.6 of \cite{LiuQuZhang},
the Principal Hierarchies of these generalized Frobenius manifolds
are related by certain linear reciprocal transformations
after adding certain families of flows generated by their Legendre vector fields,
and moreover, their topological deformations are also related by the same linear reciprocal transformations.

In this section,
we will construct two explicit linear reciprocal transformations together with Miura-type transformations that relate the RR2T of type $(2,1)$
to the bi-graded Toda hierarchy and the constrained KP hierarchy respectively.
These facts provide evidence that the RR2T of type $(2,1)$ belongs to the topological deformation of the Principal Hierarchy of $M$ \eqref{F-GAL}
up to a certain Miura-type transformation.

%
%

\subsection{Relation to the extended bi-graded Toda hierarchy}
\label{subsection: bi-graded Toda}

The RR2T of type $(2,1)$ admits a discrete symmetry,
which is analogous to Theorem 3.1 of \cite{extend AL}.
Precisely speaking, applying the wave function transformation
\[
  \phi = (\Lmd-w^+)\psi
\]
to the spectral problem \eqref{250906-1921}
of the $(2,1)$-type RR2T, we obtain
\begin{equation}
  \Lmd(\Lmd^2+u\Lmd+v)(\Lmd-w^+)^{-1}\phi = \lmd\phi.
\end{equation}
It can be verified directly that the following identity holds true:
\[
  \Lmd(\Lmd^2+u\Lmd+v)(\Lmd-w^+)^{-1}
=
  (1-w^\oplus\Lmd^{-1})^{-1}
  (\Lmd^2 + u^\oplus\Lmd + v^\oplus),
\]
where
\begin{align}
  u^\oplus &=\, u^+ + w^{+++} - \frac{(v+uw^++w^+w^{++})^+}{v+uw^++w^+w^{++}}w^+, \label{oplus-1} \\
  v^\oplus &=\, \frac{(v+uw^++w^+w^{++})^+}{v+uw^++w^+w^{++}} v, \\
  w^\oplus &=\, \frac{(v+uw^++w^+w^{++})^+}{v+uw^++w^+w^{++}}w^+.  \label{oplus-3}
\end{align}
Therefore, there is a discrete symmetry $\hat\Lmd$ of the $(2,1)$-type RR2T yielded by
\[
 (u,v,w) \mapsto (u^\oplus, v^\oplus, w^\oplus).
\]
It is clear that \eqref{oplus-1}--\eqref{oplus-3} are also the compatibility conditions between the spectral problem
\eqref{250906-1921} and the discrete evolution problem
\begin{equation}\label{discrete evolution}
  \hat\Lmd\psi = (\Lmd-w^+)\psi.
\end{equation}
Moreover, for any function $f$, we introduce the short notations
\begin{equation}
  f^\oplus := \hat\Lmd f,\quad
  f^\ominus := \hat\Lmd^{-1} f,
\end{equation}
and it can be verified that the actions of $\hat\Lmd^{-1}$ on $u,v,w$ have the form
\begin{align}
  u^\ominus &=\, u^- + w^- - \frac{v+u^-w+ww^-}{(v+u^-w+ww^-)^+}w^+, \label{ominus-1}\\
  v^\ominus &=\, \frac{(v+u^-w+ww^-)^-}{v+u^-w+ww^-}v, \\
  w^\ominus &=\, \frac{(v+u^-w+ww^-)^{--}}{(v+u^-w+ww^-)^-}w^-. \label{ominus-3}
\end{align}
\begin{rmk}\label{rmk:250915-2012}
By comparing with \eqref{PH t10}, it can be verified that for each $f=u,v,w$,
\[
  f^\oplus = f + \veps\pfrac{f}{t^{1,0}} + O(\veps^2).
\]
Hence, by the same argument as in the proof of Theorem 3.1 in \cite{extend AL}, we obtain
\begin{equation}
  \hat\Lmd = \exp\left(
    \veps\pp{t^{1,0}}
  \right),
\end{equation}
where $\pp{t^{1,0}}$ is the deformed flow mentioned in Remark \ref{rmk:extended flows}.
In other words, the discrete symmetry $\hat\Lmd$ is generated by the $\pp{t^{1,0}}$-flow in the extended RR2T of $(2,1)$-type.
\end{rmk}

If we regard the above discrete symmetry $\hat\Lmd$ as a new spatial shift operator,
then the extended RR2T of type $(2,1)$ can be related to a three-component extended bi-graded Toda hierarchy \cite{extend bi Toda}
by a linear reciprocal transformation.
To see this, we rewrite the discrete evolution problem \eqref{discrete evolution} as
\begin{equation} \label{spec-1}
  \Lmd\psi = (\hat\Lmd + w^+)\psi,
\end{equation}
then by applying $\Lmd$ to the above equation iteratively, we obtain
\begin{align}
  \Lmd^2\psi  &=\, \Lmd(\hat\Lmd+w^+)\psi = (\hat\Lmd+w^{++})(\hat\Lmd+w^+)\psi \notag\\
  &=\,
    \left(
      \hat\Lmd^2 + (w^{+\oplus}+w^{++})\hat\Lmd + w^+w^{++}
    \right)\psi, \label{spec-2}
\\
  \Lmd^3\psi
&=\,
  \left[
    \hat\Lmd^3
   +\left(
      w^{+\oplus\oplus} + w^{++\oplus} + w^{+++}
    \right)\hat\Lmd^2
  \right. \notag\\
&\quad
  \left.
    +\left(
      (w^{++\oplus}+w^{+++})w^{+\oplus} + w^{++}w^{+++}
     \right)\hat\Lmd
    +w^+w^{++}w^{+++}
  \right]\psi.   \label{spec-3}
\end{align}
On the other hand, the spectral problem \eqref{250906-1921} can be rewritten as
\[
  (\Lmd^3 + u^+\Lmd^2 + v^+\Lmd)\psi = \lmd\hat\Lmd\psi.
\]
By substituting \eqref{spec-1}--\eqref{spec-3} into the above equation, 
applying $\hat\Lmd^{-1}$ and using \eqref{ominus-1}--\eqref{ominus-3},
we obtain
\begin{equation}\label{spec problem of bigraded Toda}
  \left(
    \hat\Lmd^2 + \hat u\hat\Lmd + \hat v + \hat w\hat\Lmd^{-1}
  \right)\psi = \lmd\psi,
\end{equation}
where
\begin{equation} \label{250917-1514}
\begin{cases}
  \hat u =\, w^{+\oplus} + u+w+w^{++}, \\
  \hat v =\, v+u^-w+uw^++w^-w+ww^++w^+w^{++}, \\
  \hat w =\, v^-w + u^{--}w^-w+w^{--}w^-w.
\end{cases}
\end{equation}
Note that \eqref{spec problem of bigraded Toda} is the spectral problem of a three-component extended bi-graded Toda hierarchy
with Lax operator $\hat L$ given by \eqref{hatL}.

Based on the definition given in \cite{extend bi Toda},
the extended bi-graded Toda hierarchy corresponding to the Lax operator $\hat L$ \eqref{hatL}
consists of the following flows for all $k\geq 0$:
\begin{align}
  \varepsilon \frac{\partial \hat{L}}{\partial \hat{t}^{1,k}}
&=\,
  \frac{3}{2k!}\left[
    \left( \hat{L}^{k} (\log \hat{L}- H_k)\right)_+ , \hat L
  \right], \label{bi-graded Toda t1k}
\\
  \varepsilon \frac{\partial \hat{L}}{\partial \hat{t}^{2,k}}
&=\,
  \frac{2^{k+\frac12}}{(2k+1)!!}
 \left[
    \left( \hat{L}^{k+\frac 12}\right)_+ , \hat L
  \right],
\\
  \varepsilon \frac{\partial \hat{L}}{\partial \hat{t}^{3,k}}
&=\,
  \frac{1}{(k+1)!}\left[
    \left( \hat{L}^{k+1}\right)_+ , \hat L
  \right], \label{bi-graded Toda t3k}
\end{align}
where $H_k := 1 + \frac{1}{2} + \cdots + \frac{1}{k}$ for $k \geq 1$, and $H_0 := 0$.
The logarithm of $\hat L$ is defined by
\begin{equation}\label{log hat L}
  \log \hat L = \frac 23\log_+ \hat L + \frac 13\log_- \hat L,
\end{equation}
where the definitions of $\log_\pm\hat L$ can be found in (16)--(17) of
\cite{extend bi Toda}, and are omitted here.

\begin{rmk}
The choice of coefficients in \eqref{bi-graded Toda t1k}--\eqref{log hat L}
differs from the original definition given in \cite{extend bi Toda},
in order to ensure that the dispersionless limit of this hierarchy coincides with the Principal Hierarchy of the Frobenius manifold $\hat M$ \eqref{4.46}
with respect to the change of coordinates \eqref{bi-graded Toda change coord}.
\end{rmk}

Applying the theory of generalized Legendre transformations, see Theorem 3.13 of \cite{LiuQuZhang}, to Example \ref{ex: Legendre to bi-graded Toda},
we know that the following linear reciprocal transformation
\begin{equation} \label{time identifications 1}
  \hat t^{1,p} = t^{1,p},\quad
  \hat t^{2,p} = t^{2,p},\quad
  \hat t^{3,p} = t^{3,p},\quad p\geq 0
\end{equation}
transforms the $\{\pp{t^{1,p}}, \pp{t^{2,p}}, \pp{t^{3,p}}\}_{p\geq 0}$-flows in the Principal Hierarchy of $M$ \eqref{F-GAL}
to the Principal Hierarchy of $\hat M$ \eqref{4.46},
with respect to the change of coordinates \eqref{Legendre bigToda change}.

Moreover, notice that the dispersionless limits of the extended $(2,1)$-type RR2T and the extended bi-graded Toda hierarchy
\eqref{bi-graded Toda t1k}--\eqref{bi-graded Toda t3k} coincide with the Principal Hierarchies of $M$ \eqref{F-GAL} and $\hat M$ \eqref{4.46}
respectively,
and then, from the arguments on the spectral problems \eqref{spec-1}--\eqref{spec problem of bigraded Toda} we immediately obtain:

\begin{prop}
  Suppose functions $u,v,w$ satisfies the flows \eqref{Lax RR2T-1}--\eqref{Lax RR2T-2}
  and the deformed $\{\pp{t^{1,k}}\}_{k\geq 0}$-flows mentioned in Remark \ref{rmk:extended flows},
  then the functions $\hat u, \hat v, \hat w$ given by \eqref{250917-1514} satisfies
  \eqref{bi-graded Toda t1k}--\eqref{bi-graded Toda t3k} under the identifications \eqref{time identifications 1}.
\end{prop}

For example, by using \eqref{oplus-1}--\eqref{oplus-3}, \eqref{ominus-1}--\eqref{ominus-3},
it can be verified directly that
the $\pp{t^{2,0}}$- and $\pp{t^{3,0}}$-flows in \eqref{rad-flow-u}--\eqref{positive-flow-w}
are equivalent to the following flows
						\begin{align}
							\varepsilon \frac{\partial \hat{u}}{\partial \hat{t}^{2,0}} &=
\sqrt{2}\,(\hat\Lmd-1)\hat v+\sqrt{2}\,\hat u\frac{1-\hat\Lmd}{1+\hat\Lmd}\hat u,
							\\
							\varepsilon \frac{\partial \hat{v}}{\partial \hat{t}^{2,0}}
&=\sqrt{2}\,(\hat{\Lambda}-1)\hat{w} ,
						\quad
							\varepsilon\frac{\partial \hat{w}}{\partial \hat{t}^{2,0}} =
\sqrt{2}\,\hat{w}\frac{1-\hat{\Lambda}^{-1}}{1+\hat{\Lambda}}\hat{u} ,
\\
							\varepsilon \frac{\partial \hat{u}}{\partial \hat{t}^{3,0}} &=(\hat{\Lambda}^{2}-1)\hat{w} ,
		\quad
							\varepsilon \frac{\partial \hat{v}}{\partial \hat{t}^{3,0}} =\hat{u}\hat{\Lambda}\hat{w}-\hat{w}\hat{\Lambda}^{-1}\hat{u} ,
\quad
							\varepsilon \frac{\partial \hat{w}}{\partial \hat{t}^{3,0}} =\hat{w}(1-\hat{\Lambda}^{-1})\hat{v}
						\end{align}
in the extended bi-graded Toda hierarchy \eqref{bi-graded Toda t1k}--\eqref{bi-graded Toda t3k}
with respect to the changes \eqref{250917-1514} and \eqref{time identifications 1}.
						
\subsection{Relation to the constrained KP hierarchy}

If we regard the time variable $t^{2,0}$ in \eqref{rad-flow-u}--\eqref{rad-flow-w} as a new spatial variable $\tilde x$,
then the RR2T of type $(2,1)$ can be related to the three-component constrained KP hierarchy \cite{Cheng,cKp,Oevel-cKP}
by a linear reciprocal transformation.
To see this, we note that the $\pp{t^{2,0}}$-flow \eqref{rad-flow-u}--\eqref{rad-flow-w} in the RR2T of type $(2,1)$
is the compatibility condition between the spectral problem \eqref{250906-1921} and the following time-evolution problem
\begin{equation}
  \veps
  \pfrac{\psi}{t^{2,0}} = \sqrt{2}\,(\Lmd+b)\psi,\quad \text{where}\,\, b:=\frac{1}{1+\Lmd}(u+w).
\end{equation}
Let us rewrite the above equation as
\begin{equation}\label{ckp spec 1}
  \Lmd\psi = \left(\frac{\veps}{\sqrt{2}}\p_{\tilde x} - b\right)\psi, \quad \tilde x:= t^{2,0},
\end{equation}
then by applying $\Lmd$ to this equation iteratively, we obtain
\begin{align}
  \Lmd^2\psi
&=\,
  \left(
    \frac{\veps}{\sqrt{2}}\p_{\tilde x} - b^+
  \right)
  \left(
    \frac{\veps}{\sqrt{2}}\p_{\tilde x}-b
  \right)\psi \notag\\
&=\,
  \left[
    \frac{\veps^2}{2}\p_{\tilde x}^2
   -\frac{\veps}{\sqrt{2}}\left(b+b^+\right)\p_{\tilde x}
   +\left(bb^+-\frac{\veps}{\sqrt{2}}b_{\tilde x}\right)
  \right]\psi, \label{ckp spec 2}
\\
  \Lmd^3\psi
&=\,
  \left[
    \frac{\veps^3}{2\sqrt{2}}\p_{\tilde x}^3
   -\frac{\veps^2}{2}\left(b+b^++b^{++}\right)\p_{\tilde x}^2
  \right. \notag \\
&\qquad
  +\frac{\veps}{\sqrt{2}}
    \left(
      bb^+ + bb^{++} + b^+b^{++}
     -\frac{\veps}{\sqrt{2}}(b_{\tilde x}^+ + 2b_{\tilde x})
    \right)\p_{\tilde x} \notag \\
&\qquad
  \left.
    -\left(
      bb^+b^{++}
    -\frac{\veps}{\sqrt{2}}
     \left(
       b^+b_{\tilde x} + b^{++}b_{\tilde x} + bb^+_{\tilde x}
     \right)
    +\frac{\veps^2}{2}b_{\tilde x \tilde x}
    \right)
  \right]\psi.  \label{ckp spec 3}
\end{align}
On the other hand, the spectral problem \eqref{250906-1921} can be rewritten as
\[
  \sqrt{2}\,\left(
    \Lmd^3 + u^+\Lmd^2 + v^+\Lmd
  \right)\psi
=
  \sqrt{2}\,\lmd (\Lmd - w^+)\psi.
\]
Substituting \eqref{ckp spec 1}--\eqref{ckp spec 3} into the above equation
and using $b+b^+ = u+w$,
we obtain
\begin{equation}\label{250920-2253}
  \left(
    \frac{\veps^3}{2}\p_{\tilde x}^3
   -\frac{\veps^2}{2}\tilde z^3\p_{\tilde x}^2
   +\veps\tilde z^2\p_{\tilde x}
   +\Big(
      \tilde z^1 - \tilde z^2 \tilde z^3 + \veps\tilde z^2_{\tilde x}
    \Big)
  \right)\psi
=
  \lmd\left(\veps\p_{\tilde x} - \tilde z^3\right)\psi,
\end{equation}
where the new functions $\tilde z^1, \tilde z^2, \tilde z^3$ are given by
\begin{align}
  \tilde z^1
&=\,
  \sqrt{2}\,w^+
  \Big(
    w^+(b+b^+) + bb - b^+b^+ + v^+
  \Big)\notag \\
&\qquad
  +\veps\Big(
    2b^+b^+_{\tilde x} - 2bb_{\tilde x}
   -(b+b^+)w^+_{\tilde x}
   -2w^+(b+b^+)_{\tilde x}
   -v^+_{\tilde x}
  \Big) \notag \\
&\qquad
  +\frac{\veps^2}{\sqrt{2}}
   \Big(
     b+b^+
   \Big)_{\tilde x \tilde x}, \label{tilde z1}
\\
  \tilde z^2
&=\,
  \Big(
     w^+(b+b^+) - b^+b^+ + v^+
  \Big)
 -\frac{\veps}{\sqrt{2}}
  \Big(
    b^+_{\tilde x} + 2b_{\tilde x}
  \Big),
\\
  \tilde z^3
&=\,
  \sqrt{2}\,\Big( b + w^+ \Big). \label{tilde z3}
\end{align}
Notice that \eqref{250920-2253} is equivalent to
\begin{equation}
  \left(
    \frac{\veps^2}{2}\p_{\tilde x}^2
   +\tilde z^2
   +(\veps\p_{\tilde x}-\tilde z^3)^{-1}\tilde z^1
  \right)\psi = \lmd \psi,
\end{equation}
which is the spectral problem of the constrained KP hierarchy \cite{cKp} with Lax operator $\tilde L$ given by \eqref{KP-L}.

Based on the definition given in \cite{cKp},
the constrained KP hierarchy corresponding to the Lax operator $\tilde L$ given by \eqref{KP-L}
consists of the following flows for all $k\geq 0$:
\begin{align}\label{cKP-1}
								\veps \frac{\partial \tilde{L}}{\partial \tilde t^{2,k}}
&=\, \frac{2^{k+\frac12}}{(2k+1)!!}\left[\left(\tilde{L}^{k+\frac{1}{2}}\right)_+, \tilde{L}\right],
\\
  \veps\pfrac{\tilde L}{\tilde t^{3,k}}
&
=\,
  \frac{1}{(k+1)!}
  \left[
    \left(
    \tilde L^{k+1}
    \right)_+ , \tilde L
  \right].\label{cKP-2}
\end{align}
It can be verified directly that
\[
  \pp{\tilde t^{2,0}} = \p_{\tilde x},
\]
which is the flow given by the translation along the new spatial variable $\tilde x := t^{2,0}$.
Moreover, one can also verify that the dispersionless limits of the flows \eqref{cKP-1}--\eqref{cKP-2} coincide with the
corresponding flows in the Principal Hierarchy of the Frobenius manifold $\tilde M$ \eqref{4.51}.
Thus, in analogy with the previous subsection, we immediately obtain:

\begin{prop}
  Suppose functions $u,v,w$ satisfies the flows \eqref{Lax RR2T-1}--\eqref{Lax RR2T-2},
  then the functions $\tilde z^1, \tilde z^2, \tilde z^3$ given by \eqref{tilde z1}--\eqref{tilde z3} satisfies
  \eqref{cKP-1}--\eqref{cKP-2} under the identifications
  \[
    \tilde t^{2,k} = t^{2,k}, \qquad
    \tilde t^{3,k} = t^{3,k}
  \]
  for all $k\geq 0$.
\end{prop}

For example, if $u,v,w$ satisfies \eqref{rad-flow-u}--\eqref{positive-flow-w},
then it can be verified by straightforward calculations that
the functions $\tilde z^1, \tilde z^2, \tilde z^3$ given by \eqref{tilde z1}--\eqref{tilde z3} satisfies
the flows
\begin{align}
\frac{\partial \tilde{z}^1}{\partial \tilde t^{3,0}}
&= \tilde{z}^{3}\tilde{z}^1_{\tilde x} + \tilde{z}^{1}\tilde{z}^3_{\tilde x} - \frac{\veps}{2}\tilde{z}^1_{\tilde x\tilde x}, \qquad
\frac{\partial \tilde{z}^2}{\partial \tilde t^{3,0}}
 = \tilde{z}^1_{\tilde x},
\\
\frac{\partial \tilde{z}^3}{\partial \tilde t^{3,0}}
&= \tilde{z}^{3}\tilde{z}^3_{\tilde x} + \tilde{z}^2_{\tilde x} + \frac{\veps}{2}\tilde{z}^3_{\tilde x \tilde x}.
\end{align}
in the constrained KP hierarchy \eqref{cKP-2}, where $\tilde x := t^{2,0}$.

\begin{rmk}
  The relation between the bi-graded Toda and the constrained KP hierarchy was studied in \cite{FuYangZuo},
  which generalizes the result for the Toda hierarchy and the nonlinear Schr\"odinger hierarchy \cite{extended Toda}.
\end{rmk}

\section{Conclusions}

 In this paper, we have derived the local bihamiltonian structures of the asymmetric RR2T of types $(1,2)$ and $(2,1)$ separately, constructed a three-dimensional generalized Frobenius manifold $M$, and shown that the dispersionless limits of the flows of the $(2,1)$-type RR2T belong to the Principal Hierarchy of $M$.
 For the two hierarchies mentioned above and their corresponding bihamiltonian structures, we have established an equivalence relation via a certain Miura transformation and spatial reflection, and computed the central invariants of the associated bihamiltonian structures.

 We also provide two explicit generalized Legendre transformations that relate $M$ to the Frobenius manifold corresponding to the three-component bi-graded Toda hierarchy and the constrained KP hierarchy, and we investigate linear reciprocal transformations among these hierarchies.

 In future work, we would like to present a general theory for the asymmetric RR2T of type $(1,n)$ and $(n,1)$ for general $n\in\bbZ_+$,
 and conjecture that the central invariants of the bihamiltonian structure of the
 $(n,1)$-type RR2T are all equal to $\frac{1}{24}$.
 Furthermore, it would be interesting to study the Hamiltonian structures and properties of the symmetric RR2T of type $(n,n)$,
 which represents another generalization of the Ablowitz--Ladik hierarchy distinct from \eqref{TLO}.

\vskip 0.3cm
\noindent \textbf{Acknowledgements.}
The authors would like to thank Youjin Zhang and Si-Qi Liu for very useful discussions on this work. This work was supported by NSFC No. 12501323. The second author was supported by the ``2024 Visiting Scholar Program of Shandong University of Science and Technology'' and her work was done during her visit to Tsinghua University.


\begin{thebibliography}{99}
\bibitem{AL}
M. J. Ablowitz, J. F. Ladik, Nonlinear differential-difference equations, J. Math. Phys. 16(1975), 598--603.

							\bibitem{random permutations}
							M. Adler, P. van Moerbeke, Integrals over classical groups, random
							permutations, Toda and Toeplitz lattices, Comm. Pure Appl. Math. 54(2001), 153–205.
							
							\bibitem{Matrix integrals}
							M. Adler, P. van Moerbeke, Matrix integrals, Toda symmetries, Virasoro
							constraints and orthogonal polynomials, Duke Math. J. 80(1995),
							863911.
							
							\bibitem{GW}
							A. Brini, The local Gromov--Witten theory of $\mathbb{CP}^1$
							and integrable hierarchies, Commun. Math. Phys. 313(2012) 571–605.
							
							\bibitem{rational-reduction}
							A. Brini, G. Carlet, S. Romano, P. Rossi,
							Rational reductions of the 2D-Toda hierarchy and mirror symmetry,
							J. Eur. Math. Soc. 19(2017), 835--880.
							
							\bibitem{Brini 2012}
							A. Brini, G. Carlet, P. Rossi, Integrable hierarchies and the mirror model of local $\mathbb {CP}^1$
							, Phys. D 241(2012), 2156–2167.	
							
							
%
							
							\bibitem{extend bi Toda}
							G. Carlet, The extended bigraded Toda hierarchy, J. Phys. A 39(2006), 9411.

							\bibitem{trh -toda}
							G. Carlet, The Hamiltonian structures of the two-dimensional Toda lattice and $R$-matrices,
							Lett. Math. Phys. 71(2005), 209–226.
							
							
							\bibitem{1}
							G. Carlet, B. Dubrovin, L. P. Mertens, Inﬁnite-dimensional Frobenius manifolds for $2 + 1$
							integrable systems, Math. Ann. 349(2011), 75–115.

\bibitem{extended Toda}
G. Carlet, B. Dubrovin, Y. Zhang, The extended Toda hierarchy,
Mosc. Math. J. 4(2004), 313--332.

\bibitem{unobstructed}
G. Carlet, H. Posthuma, S. Shadrin,
Deformations of semisimple Poisson pencils of hydrodynamic type are unobstructed,
J. Differential Geom. 108(2018), 63--89.

	\bibitem{Cheng}
 Y. Cheng, Modifying the KP, the $n$-th constrained KP hierarchies and their Hamiltonian structures,
 Commun. Math. Phys. 171(1995) 661--682.

							\bibitem{TFT-2}
							R. Dijkgraaf, H. Verlinde, E. Verlinde, Topological strings in $d \leq 1$, Nucl. Phys. B 352(1991), 59--86.
							
							\bibitem{Frobenius manifold}
							B. Dubrovin, Integrable systems and classification of 2­-dimensional topological field theories,
Progr. Math. 115(1993), 313­359.
							
							\bibitem{Dubrovin1996}
							B. Dubrovin, Geometry of 2D topological field theories, In: Integrable systems and quantum groups,
Springer Berlin(1996), 120--348.
							
							\bibitem{Hamiltonian perturbation}
							B. Dubrovin, S.-Q. Liu, Y. Zhang, On Hamiltonian perturbations of hyperbolic systems of conservation laws. I. Quasi-triviality of bi-Hamiltonian
							perturbations, Comm. Pure Appl. Math. 59(2006), 559–615.
							

							\bibitem{affine Weyl groups}
							B. Dubrovin, Y. Zhang, Extended affine Weyl groups and Frobenius manifolds, Compos.
							Math. 111(1998), 167--219.

							\bibitem{Normal forms}
							B. Dubrovin, Y. Zhang,
Normal forms of hierarchies of integrable PDEs, Frobenius manifolds and Gromov–Witten invariants, eprint arXiv:math.DG/0108160.
							

\bibitem{Ferapontov}
E. V. Ferapontov, Compatible Poisson brackets of hydrodynamic type,
J. Phys. A. 34(2001), 2377.

							\bibitem{q-deformed GD}	
							E. Frenkel, Deformations of the KdV hierarchy and related soliton equations, Int. Math. Res.
							Notices (1996), 55–76.

\bibitem{FuYangZuo}
A. Fu, D. Yang, D. Zuo,
The constrained KP hierarchy and the bigraded Toda hierarchy of $(M, 1)$-type,
Lett. Math. Phys. 113(2023), 124.
							
							\bibitem{Geng}
						X. Geng, H. Dai, J. Zhu,  Decomposition of the discrete Ablowitz–Ladik hierarchy, Stud. Appl. Math. 118(2007), 281--312.

\bibitem{Kupershmidt}
J. Gibbons,  B. A. Kupershmidt, Relativistic analogs of basic integrable systems.
In: Integrable and Superintegrable Systems, World Sci.(1990), 207--231.
							
							\bibitem{AL-triham}
							S. Li, S.-Q. Liu, H. Qu, Y. Zhang,
							Tri-hamiltonian structure of the Ablowitz--Ladik hierarchy,
							Phys. D 433(2022), 133180.

\bibitem{Liu lecture notes}
S.-Q. Liu, Lecture notes on bihamiltonian structures and their central invariants,
B-Model Gromov-Witten Theory. Cham: Springer International Publishing(2019), 573--625.
							
							\bibitem{Liu2024}
							S.-Q. Liu, H. Qu, Y. Wang, Y. Zhang, Solutions of the loop equations of a class of generalized Frobenius manifolds,
Commun. Math. Phys. 405(2024), 225.


							\bibitem{GFM}
							S.-Q. Liu, H. Qu, Y. Zhang, Generalized Frobenius manifolds with non-flat unity and integrable hierarchies,
Commun. Math. Phys. 406(2025), 77.

\bibitem{LiuQuZhang}
S.-Q. Liu, H. Qu, Y. Zhang, Legendre transformations of a class of generalized Frobenius manifolds and the associated integrable hierarchies,
Commun. Math. Phys. 406(2025), 121.

							\bibitem{extend AL}
							S.-Q. Liu, Y. Wang, Y. Zhang, The extended Ablowitz–Ladik hierarchy and a generalized Frobenius
							manifold, eprint arXiv:2404.08895.

							\bibitem{GW2}
							S.-Q. Liu, D. Yang, Y. Zhang, C. Zhou, On equivariant Gromov--Witten invariants of resolved conifold with diagonal and anti-diagonal actions, Lett. Math. Phys. 112 (2022), 129.
							
							\bibitem{Bihamiltonian cohomologies}	
							S.-Q. Liu, Y. Zhang, Bihamiltonian cohomologies and integrable hierarchies I: a special case,
Commun. Math. Phys. 324(2013), 897--935.

							\bibitem{Deformations}
							S.-Q. Liu, Y. Zhang, Deformations of semisimple bihamiltonian structures of hydrodynamic type, J. Geom. Phys. 54(2005), 427–453.
							
							\bibitem{Jocabi structure}
							S.-Q. Liu, Y. Zhang, Jacobi structures of evolutionary partial differential equations, Adv. Math. 227(2011), 73–130.

							
							
%
%
%
%
							
							\bibitem{cKp}
							S.-Q. Liu, Y. Zhang, X. Zhou, Central invariants of the constrained KP hierarchies, J.
							Geom. Phys. 97(2015), 177--189.
							
							\bibitem{rtoda}
							W. Oevel, B. Fuchssteiner, H. Zhang, O. Ragnisco, Mastersymmetries, angle variables, and recursion operator of the relativistic Toda lattice, J. Math. Phys. 30(1989), 2664–2670.

\bibitem{Oevel-cKP}
W. Oevel, W. Strampp, Constrained KP hierarchy and bi-Hamiltonian structures,
Commun. Math. Phys. 157(1993), 51--81.
							
							
							\bibitem{GLt}
							I. A. Strachan, R. Stedman, Generalized Legendre transformations and symmetries of
							the WDVV equations, J. Phys. A. 50(2017), 095202.

						\bibitem{suris}
						Y. B. Suris, The problem of integrable discretization: Hamiltonian approach, in: Progress in Mathematics, Vol. 219, Birkhäuser Verlag, Basel(2003).
							
                            \bibitem{Takasaki}
							K. Takasaki, Generalized Ablowitz–Ladik hierarchy in topological string theory. J. Phys. A
							47(2014), 165201.
							
							\bibitem{3}
							K. Takasaki, T. Takebe, SDiff(2) Toda equation: Hierarchy, tau function, and symmetries,
							Lett. Math. Phys. 23(1991), 205–214.
							
							\bibitem{Toda hierarchy}
							K. Ueno, K. Takasaki, Toda lattice hierarchy, In: Group Representations and Systems of
							Differential Equations (Tokyo, 1982), Adv. Stud. Pure Math. 4, North-Holland(1984), 1–95.
							
							\bibitem{TFT-1}
							E. Witten, On the structure of the topological phase of two-dimensional gravity, Nucl. Phys.
							B 340(1990), 281–332.
							
							
							\bibitem{Central invariants}
							Y. Zhang, Central invariants of semisimple bihamiltonian structures,
in: Proceedings of the 4th International Congress of Chinese Mathematicians III,
							Higher Education Press(2008), 380–394.
							
							\bibitem{AL symmetries1}
							D. Zhang and S. Chen, Symmetries for the Ablowitz–Ladik hierarchy: Part I. Four-potential case, Stud. Appl. Math. 125(2010), 393–418.
						
						\bibitem{AL symmetries2}
						D. Zhang and S. Chen, Symmetries for the Ablowitz–Ladik hierarchy: Part II. Integrable discrete nonlinear Schr\"{o}dinger equations and discrete AKNS hierarchy, Stud. Appl. Math. 125(2010), 419–443.
							
						\end{thebibliography}
					\end{document}